\documentclass[letterpaper,12pt]{article}

\usepackage[utf8]{inputenc}
\usepackage{geometry,setspace}
\def\spacingset#1{\renewcommand{\baselinestretch}%
	{#1}\small\normalsize} \spacingset{1}
\spacingset{1.35}

\usepackage{pdfpages}
\usepackage[title]{appendix}
\usepackage{soul}

\usepackage{booktabs,caption,widetable,threeparttable}
\usepackage{tikz,fullpage}
\usetikzlibrary{arrows,	petri, topaths}
\usepackage{rotating,graphicx,lscape,xr-hyper,pgf,float,color}
\usepackage[position=top]{subfig}

\usepackage[authoryear]{natbib}
\usepackage{bibentry}
\usepackage{url}
\usepackage[colorlinks = true, allcolors = blue]{hyperref}
\usepackage[inline]{enumitem}
\usepackage{listings}
\usepackage[reftex]{theoremref}
\usepackage{amsmath, amsthm, amssymb, bm, mathtools, amsfonts, subdepth}
\mathtoolsset{showonlyrefs}

\theoremstyle{remark}

\theoremstyle{definition}

\theoremstyle{plain}
\newtheorem{thm}{Theorem}
\newtheorem{ass}{Assumption}
\newtheorem{lemma}{Lemma}

\usepackage{fixme}

\DeclarePairedDelimiter{\abs}{\lvert}{\rvert}
\DeclarePairedDelimiter{\norm}{\lVert}{\rVert}

\newcommand\inverse{^{-1}}
\newcommand{\R}{\mathbb{R}}

\newcommand{\E}{\mathbb{E}}

\DeclareMathOperator{\tr}{tr}
\DeclareMathOperator{\rk}{rk}

\title{Linear Regression with Weak Exogeneity\thanks{
We are grateful to Isaiah Andrews, Morten Ø. Nielsen, Jack Porter, and Jim Stock for useful discussions. 
Harvey Barnhard, Bas Sanders, and Chris Walker provided excellent research assistance.
}}
\author{
    \textsc{Anna Mikusheva}%
    \thanks{
        Department of Economics, M.I.T., 50 Memorial Drive, E52-526, Cambridge, MA, 02142, United States. E-mail: amikushe@mit.edu.
      }, \ 
    \textsc{Mikkel S\o lvsten}%
    \thanks{
        Department of Economics and Business Economics, Aarhus University, Fuglsangs Allé 4, Building 2621, B 13a, 8210 Aarhus V, Denmark. E-mail: miso@econ.au.dk. This research was supported by grants from the Danish National and Aarhus University Research Foundations (DNRF Chair \#DNRF154 and AUFF Grant \#AUFF-E-2022-7-3)}}
 
\usepackage{datetime}    
\newdateformat{monthyeardate}{\monthname[\THEMONTH] \THEYEAR}
\date{\monthyeardate\today}    

\begin{document}

\maketitle

\begin{abstract}
     \noindent
     This paper studies linear time series regressions with many regressors. Weak exogeneity is the most used identifying assumption in time series. Weak exogeneity requires the structural error to have zero conditional expectation given the present and past regressor values, allowing errors to correlate with future regressor realizations. We show that weak exogeneity in time series regressions with many controls may produce substantial biases and even render the least squares (OLS) estimator inconsistent. The bias arises in settings with many regressors because the normalized OLS design matrix remains asymptotically random and correlates with the regression error when only weak (but not strict) exogeneity holds. This bias's magnitude increases with the number of regressors and their average autocorrelation. To address this issue, we propose an innovative approach to bias correction that yields a new estimator with improved properties relative to OLS. We establish consistency and conditional asymptotic Gaussianity of this new estimator and provide a method for inference.
     
     \bigskip
     
     
     \bigskip
     
\end{abstract}
\clearpage

\section{Introduction}

Structural estimation in macroeconomics, finance and other economic fields studying dynamic models often employs time series data. The most used identifying assumption for structural estimation in time series settings is weak exogeneity. Weak exogeneity postulates that the structural shock has zero conditional expectation given the present and past regressor values. It is a less restrictive assumption than strict exogeneity, which additionally requires that the shocks have zero conditional expectation given future values of regressors. Strict exogeneity is implausible in most settings due to \textit{feedback}, i.e., the outcome variable in one period affects the values of the regressors in future periods.\footnote{The formalization of feedback is typically ascribed to \cite{granger1969investigating}, while \cite{engle1983exogeneity} provide an early rigorous distinction between weak and strict exogeneity. See also \cite{sims1972money,chamberlain1982general} for further discussions and an empirical example.} Specifically, if the lagged outcome variable is among the regressors, strict exogeneity cannot hold.

Another common feature of modern structural regressions is the presence of many regressors, all of which may be autocorrelated. Various motivations for the use of many regressors in time series are that the economic system generating the data is partially observed \citep{zellner1974time,wallis1977multiple}, additional controls in local projections may ensure uniformity \citep{jorda2005estimation,montiel2021local}, and long memory may be arising from an underlying high-dimensional model \citep{schennach2018long,chevillon2018generating}. See also \cite{bauwens2023we} for two applications.

We show that these two features -- weak exogeneity and many autocorrelated regressors -- can produce substantial biases and even lead to inconsistency of the ordinary least squares (OLS) estimator. The large bias in OLS may arise even when all variables are stationary (i.e., no unit roots or strong persistence is needed) and when the feedback effect violating strict exogeneity is limited to just one period. Finite sample unbiasedness of OLS relies heavily on strict exogeneity. It is well understood that OLS is biased in most time series regressions, but there is also a common belief that the biases are relatively small and of second order \citep[see, e.g.,][]{bao2007second}. Our results show that such beliefs are unwarranted and that the bias in OLS can be a first-order issue.

This paper contains several results. First, we derive a formula for the asymptotically non-negligible part of the OLS bias, explain which features of the data are responsible for it, and provide a tool to assess the potential for OLS bias in a given time series application. Second, we propose a new estimator which is asymptotically unbiased if the data contains a one-period violation of strict exogeneity. Third, we derive the asymptotic distribution for this new estimator and discuss how to conduct inferences using our new approach. Surprisingly, bias correction does not necessarily trade off with increased variance.  Based on our simulations, the standard deviations of OLS and our new estimator are very close to each other with no clear ordering. Finally, we show how our new estimator generalizes to settings where strict exogeneity is violated by multiple periods of feedback effects.

The bias of OLS arises in a setting with many regressors because the normalized design matrix, $\frac{1}{T}X'X$, remains asymptotically random even in large samples. Weak (but not strict) exogeneity allows for correlation between the randomness in the design matrix and the numerator of the OLS estimator, which leads to a bias. Specifying the feedback structure allows us to derive a formula for the leading term of the bias. Specifically, once we assume that strict exogeneity is violated by one-period feedback from the outcome variable to the next period's regressors, we show that the OLS bias is aligned with the feedback direction. The size of the bias increases with the number of regressors and their one-period ahead linear predictability.\looseness=-1

We propose a new estimator which eliminates the bias asymptotically and is consistent under the same assumptions that may lead to inconsistency of OLS. Our proposal mimics an instrumental variables (IV) estimator with an intentionally endogenous `technical' instrument: a linear combination of the regressors and their leads (future values). The main insight is that future values of the regressors in the instrument induce an endogeneity bias along the feedback direction only, the same direction along which the OLS is biased. It is therefore possible to pick the weights in the linear combination to ensure that the bias stemming from the endogenous instrument offsets the bias originating from weak exogeneity.

An important feature of our bias correction is that it is constructed based on knowledge about the regressors only. The correction is the same for any outcome variable and does not require knowledge about or estimation of the direction of the feedback mechanism. We show that the new estimator is consistent and, after proper normalization, asymptotically Gaussian when there is a one-period violation of strict exogeneity. The main results can be generalized to the case where the violation of strict exogeneity is for a finite number of periods, that is when the outcome variable has feedback effects on the regressors for $L$ periods. In such a case, the bias of OLS contains $L$ terms, corresponding to the $L$ directions of feedback.\looseness=-1

We conduct a simulation study aimed at assessing how common and how large the OLS bias is in typical macroeconomic regressions. We take a large collection of US macroeconomic indexes observed at quarterly frequency for 200 periods, extract its business cycle part, and randomly draw a regression from this data set. We show that the time series dependence and magnitude of  feedback typical for these macroeconomic data produce empirically important OLS biases. For example, in a typical regression with 25 regressors we find a bias in the feedback direction equal to half of the standard deviation, while in regression with 50 regressors this bias approximately equals one standard deviation. Depending on the number of regressors, approximately 6--21 percent of coefficients display a statistically significant difference between the OLS estimator and our proposed estimator. 
Our proposed IV-type estimator provides full correction of the bias.

Our results are related to three distinct strands of literature. First, there is a classical literature on linear equations in time series.  The exact bias of OLS in a simple auto-regression with normal errors was derived by \cite{sawa1978exact}, claiming ``the least squares estimate is seriously biased for the sample of two-digits sizes." Formulas for the second-order bias (or bias of order $1/T$) where the lagged outcome variable is the only regressor violating strict exogeneity are derived in \cite{kiviet1999alternative} and \cite{bao2007second} in a setting with a small number of regressors. A concern that weak exogeneity may lead to large biases in OLS was also raised by \cite{stambaugh1999predictive}, who considered a regression model with a very persistent (near-unit-root) regressor.  \cite{hansen2002generalized}  shows the inappropriateness of the Generalized Least Squares (GLS) estimator in linear models with only weak exogeneity as GLS mixes up the timing of observations and  leads to significant biases. The authors argue against using GLS in macroeconomic data.

The second literature is that on estimation of dynamic effects in panel data, where the presence of fixed effects (many regressors) produces a sizable bias in the coefficient on the lagged outcome variable (the weakly exogenous regressor) \citep{nickell1981biases}. Unlike the solutions proposed in that literature \citep[e.g.,][]{arellano1991some}, our solution for the time series context does not rely on the knowledge that only a single known regressor fails to be strictly exogenous nor does it require that the many regressors in the model are fixed effects for mutually exclusive groups. 
Finally, the underlying algebraic source of the bias issues, as well as some asymptotic statements related to Gaussianity of quadratic forms, are connected to the issues arising in linear models with many instruments and/or many regressors -- see \cite{hansen2008estimation,chao2012asymptotic,kline2020leave}.

The rest of the paper is organized as follows. Section \ref{sec: bias OLS} derives the formula for the leading term of the OLS bias, provides intuition for the findings, and discusses features of the data responsible for the bias. Section \ref{sec: IV estimator} introduces a new asymptotically unbiased estimator under the assumption of one-period feedback and establishes its consistency. Section \ref{sec: inference} establishes asymptotic Gaussianity of the newly proposed estimator and suggests a valid inference procedure. Section \ref{sec: multi period} extends some of these results to settings with multi-period feedback. Section \ref{sec: mcs} contains simulation studies assessing the empirical relevance of the discussed issues in typical macroeconomic data sets. All proofs are in Appendix \ref{sec: proofs}.

\paragraph{Notation:} For any vector $x$, $\|x\|=\sqrt{x'x}$ gives the $L_2$ norm of $x$. For any matrix $A$ (not necessarily square) $\rk(A)$ is the rank of $A$, $\|A\|=\sup_x\frac{\|Ax\|}{\|x\|}$ denotes the operator norm (the positive square root of the largest eigenvalue of $A'A$), and $\|A\|_F=\sqrt{\tr(A'A)}$ denotes the Frobenius norm. Throughout, we let $c < 1$ and $C$ be strictly positive and finite (generic) constants, that may have different values in different equations, but do not depend on the sample size $T$. The $q\times q$ identity matrix is $I_q$ while $I=I_T$. 

\section{Inconsistency of OLS}
\label{sec: bias OLS}

\subsection{Model and assumptions}

Consider a linear time series regression
\begin{align}
	y_t &= x_t'\beta + \varepsilon_t, && t \in \{ 1,\dots,T \},
\end{align}
where the regressors $x_t \in \R^K$ are weakly exogenous and the number of regressors $K$ is large. We model this by assuming that $K$ may diverge proportionally to $T$ or slower. All features of the data generating process are implicitly indexed by $T$, but we drop this index for compactness of notation. The object of interest is the linear contrast $\theta=r'\beta$ for a known (non-random) $K$ vector $r$. A leading case is $r'\beta = \beta_1$.

The assumption of weak exogeneity imposes:
\begin{align}
    \E[ \varepsilon_t \!\mid\! x_t,x_{t-1},\dots ]&=0.
\end{align}
Weak exogeneity is considerably less restrictive than strict exogeneity, which assumes that $\E[ \varepsilon_t \!\mid\! X ]=0$, with $X$ denoting the $T\times K$ set of all regressors (including past, present, and future). In the context of time series economic data, strict exogeneity is seldom plausible. In contrast, weak exogeneity allows the structural shock to influence future regressor values through feedback. Presence of such feedback is very likely as variables employed in macro estimation often evolve as a joint dynamic process and affect each other either instantaneously (as $x_t$  affects $y_t$) or with some lag. When lagged outcome variable are incorporated as regressors, strict exogeneity is certainly violated. For further discussion regarding the plausibility of weak and strict exogeneity in various applications, see \citet[][ch. 16]{stock2019introduction}.\looseness=-1

Standard arguments for unbiasedness of the OLS estimator $r'\hat{\beta}^\text{OLS}=r'(X'X)^{-1}X'y$ rely heavily on strict exogeneity. Although it is well-known that OLS exhibits bias under weak exogeneity \citep[see, e.g.,][chapter 8.2]{hamilton1994time}, a common belief persists that this bias is small and asymptotically negligible. Here we claim that this widespread belief is both incorrect and misleading: the bias of OLS can be substantial, potentially leading to inconsistency. \looseness=-1

Standard statements regarding OLS consistency in a time series context hinge on the assumption that the normalized design matrix concentrates around its expectation: $\norm{\frac{1}{T}X'X - Q} \xrightarrow{p} 0$, where $Q$ is non-singular \citep[see, e.g.,][Assumption 8.6]{hamilton1994time}. For example, \cite{gupta2023robust} assumes that $K^3/T\to 0$ to invoke a Law of Large Numbers for the normalized design matrix. However, the normalized design matrix remains asymptotically random when the number of regressors (and thus the design matrix dimensionality) grows fast enough. Under strict exogeneity, the randomness of $\frac{1}{T}X'X$ does not cause a problem as one can condition on $X$ in the analysis, thus treating the design matrix as non-random. This approach is infeasible under weak exogeneity. In such instances, $\frac{1}{T}X'X$ is not just a random matrix, but is also correlated with $\frac{1}{T}X'y$. This correlation produces a bias that may become the leading term in the asymptotic analysis of OLS.

The size and form of the OLS bias depend on the feedback mechanism, or how past errors in the outcome variable affect future regressor values. We start by assuming one-period linear feedback, i.e., the present error term $\varepsilon_t$ only affects the regressors in the subsequent period, namely $x_{t+1}$. Section \ref{sec: multi period} extends our results to feedback lasting a finite number of periods.

\begin{ass}\label{Ass: OLS}
    \begin{enumerate}[label=(\roman*)]

	\item The observed regressors $x_t$ can be decomposed as $x_t = \tilde x_t + \alpha \varepsilon_{t-1},$ where the $T\times K$ matrix $\tilde{X} = [\tilde x_1,\dots,\tilde x_T]'$ has full rank. \label{Ass: OLS 1}
		
	\item The errors $\{\varepsilon_t\}_{t=0}^T$ are i.i.d. conditionally on $\tilde {X}$ with $\E[\varepsilon_t|\tilde X]=0$, $c<\sigma^2 := \E[\varepsilon_t^2|\tilde X] <C$, and $\E[\varepsilon_t^4|\tilde X]<C$ almost surely. \label{Ass: OLS 2}
		
	\item The size of the non-random vectors $\alpha, r \in \R^K$ relative to the strictly exogenous design matrix $\frac{1}{T}\tilde X'\tilde X$ is bounded: $\alpha' (\frac{1}{T}\tilde X'\tilde X) \inverse \alpha = O_p(1)$ and $r'(\frac{1}{T}\tilde X'\tilde X)^{-1}r = O_p(1)$.\label{Ass: OLS 3}
		
	\item The number of regressors $K$ may diverge with sample size $T$ such that $\frac{K}{T}<1-c$. \label{Ass: OLS 4}
		
\end{enumerate}
\end{ass}

Part (i) describes a specific violation of strict exogeneity where the present error term only affects the regressors in the subsequent period. If $\alpha=0$, then all regressors $x_t=\tilde{x}_t$ are strictly exogenous. Section \ref{sec: multi period} generalizes this assumption to allow for multi-period feedback and it can be generalized further to infinite-order feedback as long as the feedback magnitude decays sufficiently fast.\footnote{The results are available from the authors upon request.} Such a generalization covers the case with a lagged outcome variable used as a regressor as it leads to an infinite geometrically decaying feedback. The usual AR(p) model is a  special case of this. Linearity of the feedback is a natural starting point that arises if, for example, one assumes that all regressors and errors are jointly Gaussian (an assumption that originally motivated the OLS estimator). Furthermore, the wast majority of models estimated using macroeconomic data are linear due the limited length of time series data.\looseness=-1

The amount of assumptions imposed on the strictly exogenous part of the regressors $\{\tilde x_t\}_{t=1}^T$ is kept to a minimum, the main part of which is a full rank condition. Specifically, we do not assume stationarity or impose any moment conditions on $\tilde X$. Such generality is possible primarily because our analysis is done conditionally on $\tilde X$. One could impose additional assumptions on the evolution of $\tilde x_t$ and potentially use them to improve efficiency of estimators. However, such assumptions increases the potential for mis-specification and severely reduces the applicability of the results.  

Part (ii) is a standard set of assumptions on the error terms in homoskedastic regression models. Allowing for auto-correlation of error terms is a possible generalization, though we leave it to future research.  For such generalization, the main attention should be placed on a distinction between the feedback versus auto-correlation. Specifically, one could decompose the regression error into a sum of two parts: one that produces feedback but no autocorrelation (as in the current paper)  and one that auto-correlates but does not produce any feedback. The feedback leads to a bias in OLS while the autocorrelated part does not cause bias but rather necessitates different formulae for standard errors of linear estimators.

Part (iii) places a very loose bound on the magnitudes of $\alpha$ and $r$ relative to the scaled design matrix $\frac{1}{T}\tilde X'\tilde X$. This condition is, for example, satisfied if $\norm{\alpha}$ and $\norm{r}$ are bounded and the smallest eigenvalue of $\frac{1}{T}\tilde X'\tilde X$ is separated from zero.
Part (iv) allows a wide range for the number of regressors, including a small fixed number. It requires that the degrees of freedom $T-K$ diverges to infinity as the sample size increases. 

The violation of strict exogeneity in Assumption \ref{Ass: OLS} seems minimal as it is solely due to feedback from the dependent variable $y_t$ to the one-period-ahead regressors $x_{t+1}$ and the magnitude of the feedback is bounded. However, this violation is enough to produce inconsistency of the OLS estimator.

\subsection{Parameter estimator}

In order to provide some intuition for how and why the OLS bias arises, we consider a special case. Suppose, therefore, that only the first regressor fails to be strictly exogenous. Then, the first regressor experiences a one-period feedback: $x_{1t}=\tilde{x}_{1t} + a\varepsilon_{t-1}$, while the $T\times (K-1)$ matrix of all other regressors, $X_{-1}=\tilde{X}_{-1}$, is strictly exogenous. Let $X_1$ be the $T\times 1$ vector with elements $x_{1t}$ and $\tilde X_1$ its strictly exogenous part. Furthermore, we take $\tilde X$ as fixed and suppose that $\frac{1}{T}\tilde X'\tilde X=I_K$. 

In this special case, we now derive the bias of the OLS estimator for the first coefficient $\beta_1$. According to the Frisch-Waugh-Lovell theorem, we have
\begin{align}
    \hat{\beta}_1^\text{OLS}=\frac{X_1'M_{-1}y}{X_1'M_{-1}X_1}
    \quad \text{and} \quad
    \hat\beta_1^\text{OLS}-\beta_1=\frac{X_1'M_{-1}\varepsilon}{X_1'M_{-1}X_1},
\end{align}
where the partialling out operator $M_{-1}=I-\tilde X_{-1}(\tilde X_{-1}'\tilde X_{-1})^{-1}\tilde X_{-1}'$ is the projection on the space orthogonal to $\tilde{X}_{-1}$. Notice that $M_{-1}=\big(M_{st}^{*}\big)$ is fixed (or can be conditioned on) so that all randomness comes from $\varepsilon$. The denominator of $ \hat\beta_1^\text{OLS}-\beta_1$ has a standard asymptotic behavior and once normalized by $\frac{1}{T}$ converges to a non-random and non-zero limit.

If all regressors were strictly exogenous ($X_1=\tilde X_1$), then the OLS estimator would have been properly centered. However, $X_1$ contains randomness that is correlated with $\varepsilon$. This leads to a non-zero expectation of the numerator in $\hat\beta_1^\text{OLS}-\beta_1$: 
\begin{align}
    \E[X_1'M_{-1}\varepsilon \mid \tilde X]=\E[\tilde X_1'M_{-1}\varepsilon\mid \tilde X]+a \E[\sum_{s,t}M^{*}_{st}\varepsilon_{s-1} \varepsilon_t\mid \tilde X]= a \sigma^2 \sum_t M^{*}_{tt-1}.
\end{align}
While the weakly exogenous regressor $x_{1t}$ is not correlated with the contemporaneous error $\varepsilon_t$, the process of partialling out the remaining regressors mixes up the timing of the observations and makes the partialled out regressor correlated with the contemporaneous error.  This derivation suggest a form for the leading term of the bias of $\hat{\beta}_1^\text{OLS}$. We may notice that the bias depends on $a$ which governs the magnitude of feedback in the regressors and on the trace of the lower diagonal of the projection matrix $M_{-1}$ ($\sum_t M^{*}_{tt-1}$). As discussed below, this lower trace may be on the order of the number of regressors and increases with their average autocorrelation. It is also possible to show that the remaining OLS coefficients (e.g., $\hat\beta_2^\text{OLS}$) have an asymptotically negligible bias (although they have a finite sample bias). Thus for any linear combination $\theta=r'\beta$, the bias depends on the weight placed on $\beta_1$.


The insight from the preceding special case can be extended. In fact, the OLS estimator of a linear contrast remains invariant under linear transformations of the regressors. For instance, if we perform a regression of $Y$ on $XA$, where $A$ is a $K\times K$ matrix with full rank that linearly transforms the regressors, then the OLS serves as an estimator of $A^{-1}\beta$. To achieve the contrast $\theta$, we should apply $A'r$ as the weighting.

Any OLS scenario simplifies to the one discussed earlier when we choose $A=(\frac{1}{T}\tilde X'\tilde X)^{-1/2}$, with the square root selected to ensure that $A'\alpha$ is proportional to the first basis vector. These insights, coupled with certain technical derivations, lead to the following characterization of the OLS bias.

\begin{thm}[Inconsistency of OLS estimator]\label{thm- OLS}
    Suppose Assumption \ref{Ass: OLS} holds. Then,
    \begin{align}
        r'\hat \beta^\emph{OLS} - r'\beta = \sigma^2 r' \bar S^{-1} \alpha \sum\nolimits_{t=2}^{T} \tilde M_{tt-1} + o_p(1),
    \end{align}
    where $\bar S = \tilde X'\tilde X + \alpha \alpha' \sigma^2 (T-K)$ and $\tilde M = I - \tilde X(\tilde X'\tilde X)^{-1} \tilde X'$.
\end{thm}

Theorem \ref{thm- OLS} presents the leading term of the OLS bias under Assumption \ref{Ass: OLS}. Note that this formula uses the lower diagonal trace $\sum_t \tilde M_{tt-1}$ of the projection matrix orthogonal to $\tilde X$ instead of $\sum_t M^{*}_{tt-1}$. Theorem \ref{thm- OLS} establish asymptotic negligibility of this difference.

It is worth discussing the size of the lower trace of the projection matrix as well as the size of the biases we may see in applications. Consider the term $\sum_{t}\tilde{M}_{tt-1}=-\sum_t\tilde{P}_{tt-1}$, where $\tilde{P}=\tilde{X}(\tilde{X}'\tilde X)^{-1}\tilde X'$. Note that this quantity is unchanged by any full rank rotation of the regressors. For simplicity, suppose that the regressors are generated by a stationary and ergodic process with $\frac{1}{T}\tilde{X}'\tilde X=I_K$. Then 
\begin{align}
    \sum_{t}\tilde{M}_{tt-1}=-\frac{1}{T}\sum_{t}\tilde{X}_t'\tilde{X}_{t-1}\approx -\E\tilde{X}_t'\tilde{X}_{t-1}.
\end{align}
Thus, the lower diagonal trace of $\tilde M$ measures a linear connection between $\tilde{X}_t$ and $\tilde X_{t-1}$. We can also notice that $\frac{1}{T}\sum_{t}\tilde{X}_t'\tilde{X}_{t}= \E\tilde{X}_t'\tilde{X}_{t}=K$, since we are dealing with $K$-dimensional vectors. In time series settings we tend to work with data where the regressors are autocorrelated. This way we should expect $\sum_{t}\tilde{M}_{tt-1}\approx -\rho K$, where $\rho$ is the average (over different regressors) scalar measure of regressor predictability. 

We can notice that $\bar S \geq \tilde X'\tilde X $ which leads to Assumption \ref{Ass: OLS}\ref{Ass: OLS 3} implying $r' \bar S^{-1} \alpha=O_p(1/T)$. Let us momentarily assume $Tr' \bar S^{-1} \alpha$ converges to a constant $r_\alpha$. When $K$ grows proportionally to the sample size while $\rho$ remains bounded away from zero, the leading bias term becomes approximately $\sigma^2 r_\alpha\cdot\frac{\rho K}{T}$. This renders the OLS estimator for $r'\beta$ inconsistent. Notably, the potential bias term is more pronounced for regressors with higher first-order autocorrelation and a larger number of regressors. 

Although $\frac{1}{T}\sum\nolimits_{t=2}^{T} \tilde M_{tt-1}$ (the lower diagonal trace of $\tilde M$ over the sample size) is typically unavailable in practice as the strictly exogenous part of the regressors is unobservable, the asymptotically equivalent lower diagonal trace of $M=I-X(X'X)^{-1}X'$ over the sample size can be calculated from available data. If this quantity is non-negligible in an application, even a small violation of strict exogeneity may result in substantial biases.

Lastly, the extent of bias depends on the alignment between the contrast $r$ and the feedback direction $\alpha$, making the feedback direction the most affected contrast direction. Contrasts in all directions orthogonal to $\alpha$ (using the scalar product weighted by $\bar S^{-1}$) experience only negligible bias. In our special case where only the first regressor is weakly exogenous, this observation corresponds to the OLS estimator of $\beta_1$ being the sole estimator with significant bias. Thus, for any linear contrast $r'\beta$, the bias depends on the weight placed on $\beta_1$. In real applications, the feedback direction is unknown and challenging to estimate empirically, given that $\alpha$ is a $K\times 1$ vector.

\paragraph{Simulations} We demonstrate the potential for OLS bias through a small scale simulation that varies the number of regressors and their short-term dependence. Data is generated following Assumption \ref{Ass: OLS}. The outcome vector is generated as $y = X\beta + \varepsilon$ with $\varepsilon \sim N(0,I)$ and $\beta=0$. The design matrix is generated as $x_{1t} = \tilde x_{1t} + a\varepsilon_{t-1}$ and $X_{-1}=\tilde X_{-1}$, where $\tilde X$ is generated as a rotated VAR(1) process with $\tilde X\tilde X'/T=I_K$, independent from $\varepsilon$. Specifically, we generate $V_t=\rho V_{t-1}+u_t$  with $\{u_t\}_{t=1}^T$ \emph{i.i.d.} $N(0,I_K)$ and define $\tilde X=V(V'V/T)^{-1/2},$ where the square root comes from Cholesky decomposition. Across simulations, we fix the sample size at $T=200$ and the coefficient on the feedback mechanism at $a=1.5$. Simulation results are summarized in Figure \ref{fig:bias} with the left panel showing results for the number of regressors $K$ between $4$ and $100$ (fixing $\rho$ at $0.8$).\footnote{All results are presented as sixth order polynomial fits to the actual results across $K$.} The right panel reports the results for the auto-correlation in regressors $\rho$ between $0$ and $0.98$ (fixing $K$ at $50$). We report simulated values of absolute bias and standard deviation for the first coordinate of OLS together with the mean absolute value of the ratio of the lower trace of $M$ to the sample size. Additionally, we report the same summary statistics for the new estimator proposed in Section \ref{sec: IV estimator}.

\begin{figure}[htbp]
\centering
\includegraphics[width=\columnwidth]{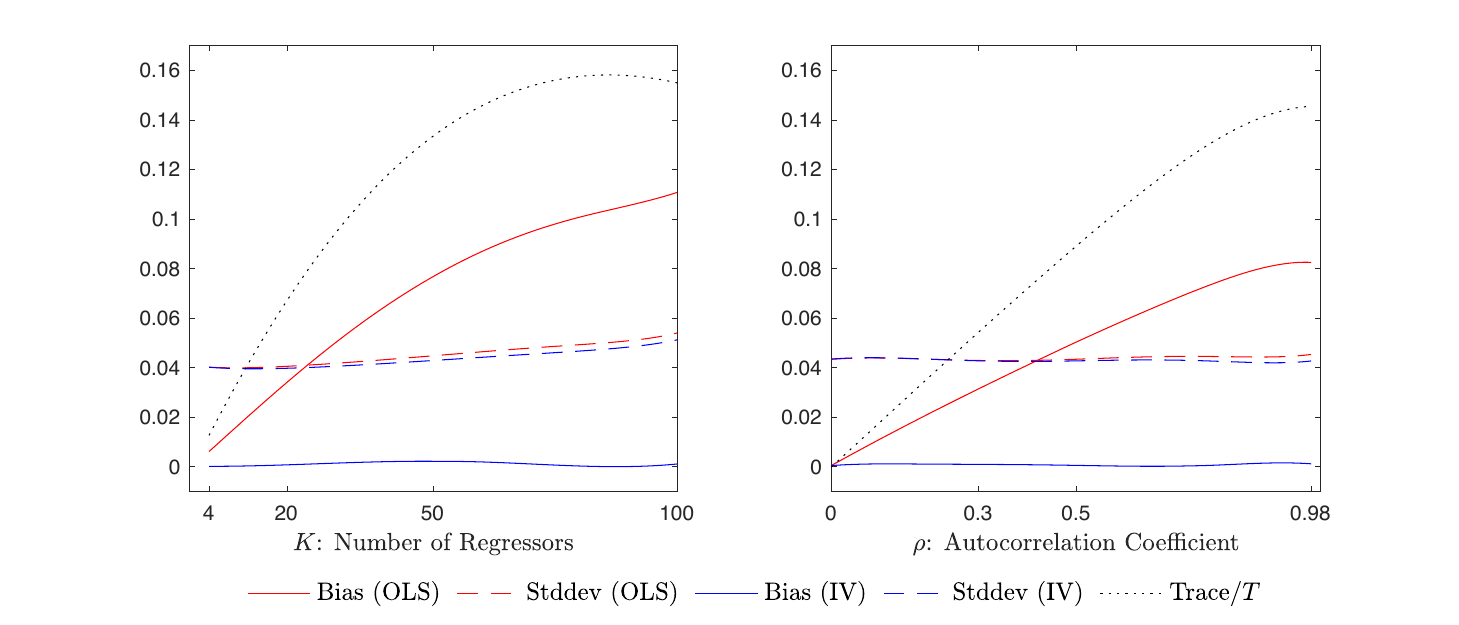}
\caption{Absolute Bias and Standard Deviation of OLS and IV with $T=200$}
\label{fig:bias}
\end{figure}

The results presented in Figure \ref{fig:bias} suggest that the bias of the OLS estimator can easily surpass its standard deviation, leading to highly unreliable statistical inferences. For example, even a regression with just 20 regressors may exhibit an OLS bias comparable in magnitude to the standard deviation in a sample of size 200 --- a very common setting in macroeconomic applications. The observable lower trace of $M$ divided by the sample size provides a highly predictive measure of the magnitude of the bias in the most affected direction. Notably, both the bias in the most affected direction and the lower diagonal trace tend to increase with the number of regressors and the first auto-correlation of the regressors.\footnote{
One may be worried that the observed biases are due to persistence in the regressors as the coefficient $\rho$ in an AR(1) process measures not only short-term dependence but also the long-run persistence. We re-run simulations generating $V_t = \rho u_{t-1} + u_t$ as an MA(1) process. The results are presented in Appendix \ref{app: simulations} and are essentially identical to those reported in Figure \ref{fig:bias}. }

Figure \ref {fig:bias4} presents results for the same simulation design but with a sample size of $T=800$, while the number of regressors varies from 16 to 400. Here, we observe that the bias reaches the same level as in Figure \ref{fig:bias} when the number of regressors is the same fraction of the sample size, while the standard deviations drops two-fold. This demonstrates the inconsistency of the OLS for the worst direction when the number of regressors $K$ grows proportionally to $T$. In essence, the estimator concentrates around an incorrect value as the sample size increases.

\begin{figure}[htbp]
\centering
\includegraphics[width=\columnwidth]{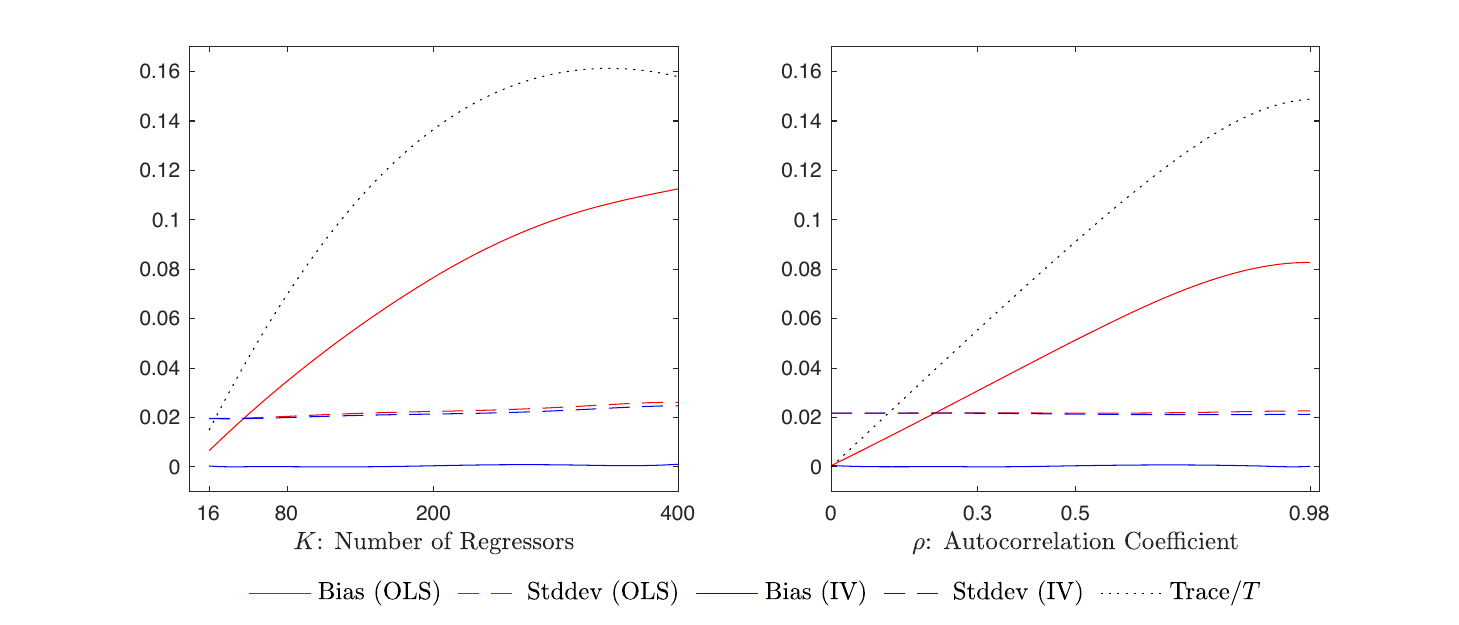}
\caption{Absolute Bias and Standard Deviation of OLS and IV with $T=800$}
\label{fig:bias4}
\end{figure}

\subsection{Variance estimator}

In order to make reliable statistical inferences in a linear regression we usually need an estimator for the variance of the error term $\sigma^2$. It is well known that in a regression with many regressors and strict exogeneity one has to adjust properly for the degrees of freedom in order to correct for over-fitting. The most commonly applied OLS estimator of the error variance uses this adjustment $\hat\sigma^2=\frac{y'My}{T-K}$. It is not obvious ex ante whether this estimator retains consistency with many regressors and only weak exogeneity. On the one hand, a large bias of the OLS estimator for coefficients raises concerns about the consistency of the variance estimator. On the other hand, the bias arises only in the direction of the feedback but not in all directions orthogonal to the feedback, which are numerous. Theorem \ref{them: inconsistency of variance} answers the question of the consistency of the OLS variance estimator.

\begin{thm}[Inconsistency of OLS variance]\label{them: inconsistency of variance}
    Suppose Assumption \ref{Ass: OLS} holds. Then,
    \begin{align}
        \frac{\hat \sigma^2}{\sigma^2} = 1 - \frac{\sigma^2 \alpha' \bar S^{-1} \alpha}{T-K} \left( \sum\nolimits_{t=2}^{T} \tilde M_{tt-1}\right)^2 + o_p(1).
    \end{align}
\end{thm}

Theorem \ref{them: inconsistency of variance} states that the OLS estimator of the error variance is biased downward and reports an overly optimistic measure of fit in a setting with one-period feedback and many regressors. The size of the bias can be judged by the trace of the lower diagonal of the projection matrix $M$. When the number and predictability of the regressors are high enough to imply that the OLS estimator of the linear contrast in the worst direction (the feedback direction) is inconsistent, the OLS estimator of variance is inconsistent as well. Despite this inconsistency result, the analogues biases we observe in simulations tend to be relatively minor.\looseness=-1

\section{A consistent IV estimator}
\label{sec: IV estimator}

\subsection{The idea of the proposed estimator}

Let us introduce a $T\times T$ shift matrix $D$ (or lag operator matrix) that shifts the time series index back by one; its only non-zero elements are $D_{t,t-1}=1$ for all $t$. The transpose $D'$ is the lead operator, moving the time-series index forward by one. Notice that the lower trace appearing in the bias of OLS can be written as $\sum\nolimits_{t=2}^{T} \tilde M_{tt-1} = \tr(D'\tilde M)$. Consider also a $T\times T$ matrix $\Gamma$ measurable with respect to $\tilde X$ such that $\|\Gamma\|\leq 1-c$.

We propose an IV-inspired estimator that relies on an endogenous instrument $Z=(I-\Gamma')X$: 
\begin{align}
    \hat{\beta}^\text{IV} (\Gamma)= (Z'X)\inverse Z'y=(X'(I-\Gamma )X)^{-1}X'(I-\Gamma )Y.
\end{align}
The key insight behind the estimator is that we use a deliberately invalid instrument created in a way so the `invalid instrument' bias offsets the bias in OLS. In order to eliminate the asymptotic bias we require that $\Gamma$ solves the following (non-linear) one-dimensional equation:
\begin{align}\label{eq:qform1}
    \tr\!\big[ D'(I-\Gamma )\tilde M_\Gamma \big]=0
\end{align}
where $\tilde M_\Gamma=I-\tilde{X}(\tilde X'(I-\Gamma )\tilde X)^{-1}\tilde X'(I-\Gamma )=I-\tilde X(\tilde Z'\tilde X)^{-1}\tilde Z'$ is an oblique projection off $\tilde X$ in the direction of $\tilde Z=(I-\Gamma')\tilde X$.
Further below, we establish that the new estimator $\hat \beta^\text{IV}(\Gamma)$ is a properly centered $\sqrt{T}$ asymptotically Gaussian (conditionally on $\tilde X$) estimator under Assumption \ref{Ass: OLS}. The main theoretical result about bias is obtained for a general $\Gamma$ with the leading example of $\Gamma=\gamma D$ for $|\gamma|<1-c$.

Let us give some intuition of why this approach would work. Consider again the special case when all regressors but the first are strictly exogenous. Namely $x_{1t}=\tilde x_{1t}+a \varepsilon_{t-1}$, while $X_{-1}=\tilde{X}_{-1}$ and consider the first coefficient $\beta_1$ only, as it is the most biased direction. Consider also the case of $\Gamma=\gamma D$, thus $z_t=x_t-\gamma x_{t+1}$. In this setting all but the first instruments are strictly exogenous, while the first instrument $z_{1t}=\tilde z_{1t}+\varepsilon_{t-1}-\gamma\varepsilon_t$ is endogenous as it correlates with the contemporaneous regression error. The direction of the instrument endogeneity in general coincides with the feedback direction, as the transformation $\Gamma$ mixes the timing of the observations but preserves the feedback direction.

Oblique projections underlie the geometry of IV estimation similarly to how orthogonal projections explain the geometry of OLS. An oblique projection is defined as $M_{Z,X}=I-X(Z'X)^{-1}Z'$, where $X$ and $Z$ are of the same dimension and $Z'X$ is invertible. Oblique projections satisfy idempotency, $M_{Z,X}^2=M_{Z,X}$, but not symmetry, $M_{Z,X}' \neq M_{Z,X}$. One can easily show that the Frisch-Waugh-Lowell theorem holds for oblique projections as well. Specifically, let $Z=[Z_1;Z_{-1}]$ where the dimensions of the corresponding Z's and X's coincide and all proper matrices are invertible. Then 
\begin{align}
    \hat\beta_1^\text{IV} (\Gamma) = \frac{Z_1'M_{Z_{-1},X_{-1}}Y}{Z_1'M_{Z_{-1},X_{-1}}X_1} 
    \quad \text{and} \quad
    \hat\beta_1^\text{IV} (\Gamma)-\beta_1=\frac{Z_1'M_{Z_{-1},X_{-1}}\varepsilon}{Z_1'M_{Z_{-1},X_{-1}}X_1}.
\end{align}
Notice that since $X_{-1}$ is strictly exogenous, $Z_{-1}$ is strictly exogenous as well. If we condition on $\tilde X$ we may then treat $M_{Z_{-1},X_{-1}}$ as fixed. 
We look at the numerator:
\begin{align}\label{eq: numerator IV}
\E [Z_1'M_{Z_{-1},X_{-1}}\varepsilon]=\E[\tilde{Z}_1'M_{Z_{-1},X_{-1}}\varepsilon]+ a\E[\sum_{s,t} M^*_{st}(\varepsilon_{s-1}-\gamma \varepsilon_s)\varepsilon_t]= a\sigma^2\sum_t(M_{tt-1}^*-\gamma M_{tt}^*).
\end{align}
Here we used $M_{Z_{-1},X_{-1}} = (M^*_{st})$ for shortness of notation. Our goal is to choose $\gamma$ in a manner that renders the last sum equal to zero. This should generally be feasible, given that the diagonal elements of projection matrices tend to dominate those on the lower diagonal. In the proof, we show that changing $M_{Z_{-1},X_{-1}}$ to $\tilde M_\Gamma$ in the bias expression introduces only an asymptotically negligible difference. Thus, the expectation of the numerator is asymptotically equivalent to $a\sigma^2\tr[D'(I-\Gamma)\tilde M_\Gamma]$. By selecting $\Gamma$ to solve equation \eqref{eq:qform1}, we therefore achieve an asymptotically unbiased estimator.

\subsection{Consistency of estimator}

Let $\gamma_0$ denote the solution to \eqref{eq:qform1} among matrices $\Gamma_0 = \gamma_0 D$. Since we only observe regressors $X$, knowing $\tilde{X}$ is equivalent to knowing the feedback direction $\alpha$. This makes solving equation (\ref{eq:qform1}) infeasible in practice. Let $\hat \gamma $ and $\hat\Gamma=\hat\gamma D$ be the solution to the empirically feasible equation:
\begin{align}\label{eq:qform2}
    \tr\!\big[ D'(I-\hat\Gamma) M_{\hat\Gamma} \big] =0 
    \quad \text{where} \quad 
    M_{\hat\Gamma}=I-{X}( X'(I-\hat\Gamma ) X)^{-1} X'(I-\hat\Gamma ).
\end{align}
The following Lemma gives sufficient conditions for existence and uniqueness of a solution to equation \eqref{eq:qform1}.

\begin{lemma}\label{lem: solution exists}
    Suppose that $\tilde X$ is a $T\times K$ matrix of rank $K$ and $\Gamma=\gamma D$.
    \begin{enumerate}[label=(\roman*)]
        \item If $K < T/5$, then there exists a unique $\gamma \in [-1/2,1/2]$ solving equation \eqref{eq:qform1}. \label{lem: solution exists 1}
        
        \item If $|\tr(D'\tilde M)|\leq \mu^2 K$ and $K < T/[1+(1+\mu)^2]$ for some $\mu \in [0,1]$, then there exists a unique $\gamma$ solving equation \eqref{eq:qform1} such that $ |\gamma|<\mu/(1+\mu)$. \label{lem: solution exists 2}
    \end{enumerate}
\end{lemma}

One may state an analog of Lemma \ref{lem: solution exists} for equation \eqref{eq:qform2} as well due to their analog structure. According to Lemma \ref{lem: solution exists}, equation \eqref{eq:qform1} typically can be solved in a one-dimensional family of transformations $\Gamma=\gamma D$. It is possible to search for a solution in other classes of matrices; we leave the question of finding an optimal class of transformations $\Gamma$ to future research. The proof of Lemma \ref{lem: solution exists} reveals that though equation \eqref{eq:qform1} is non-linear, its solution can be found relatively fast as a fixed point of a contraction.

One may think that $K<T/5$ in Part (i) of Lemma \ref{lem: solution exists} is a restrictive condition. It arises because we do not put any restriction on $\tilde{X}$ other than full rank. This accommodates an extensive range of idiosyncrasies in the generation of regressors, surpassing those encountered in typical macroeconomic time series data sets. Placing some restrictions on regressors may weaken the restriction on the sample size significantly. For example, putting a bound on the average auto-correlation of the regressors, $\mu$, eases this restriction considerably as shown in Part (ii) of Lemma \ref{lem: solution exists}. We also notice that if the original potential for bias is small, then the $\gamma$ that removes the bias is small as well. Specifically, if $\tr[ D'\tilde M ]=0$ so the original OLS has no (asymptotic) bias, our estimator defaults back to OLS.

The following theorem establishes the asymptotic bias of IV for a generic choice of $\Gamma$ and shows consistency for the specific choice of $\hat \Gamma = \hat \gamma D$. 

\begin{thm}\label{Thm: IV}
    Suppose Assumption \ref{Ass: OLS} holds.
    \begin{enumerate}[mode=unboxed, label=(\roman*)]
        \item If $\Gamma$ is $\tilde X$-measurable and $\|\Gamma\|<1-c$, then
        \begin{align}
            r'\hat \beta^\emph{IV}(\Gamma) -r'\beta = \sigma^2r'\bar{S}_\Gamma^{-1}\alpha \tr\!\big[ D'(I-\Gamma)\tilde M_\Gamma \big] + o_p(1),
        \end{align} 
        where $\bar S_\Gamma = \tilde X'(I-\Gamma)\tilde X + \sigma^2 \alpha \alpha' \tr\!\big[(I-\Gamma )\tilde M_\Gamma \big]$. Specifically, it follows that $r'\hat \beta^\emph{IV}(\Gamma)$ is consistent for $r'\beta$ when $\Gamma$ solves equation \eqref{eq:qform1}. \label{Thm: IV 1}

        \item If $K < T/5$, then 
        $\hat\gamma -\gamma_0 = O_p(\frac{1}{T}),$ and $r'\hat \beta^\emph{IV}(\hat\Gamma) -r'\hat \beta^\emph{IV}(\Gamma_0) =o_p(\frac{1}{\sqrt{T}})$. \label{Thm: IV 2}
    \end{enumerate}
\end{thm}

The result of Theorem \ref{thm- OLS} is a special case of Part (i) of Theorem \ref{Thm: IV} for $\Gamma=0$. Using $\Gamma$ that solves equation (\ref{eq:qform1}) produces a consistent estimator for any reasonable contrast, and the proposed solution does not depend on the contrast of interest. The ideal $\Gamma_0$ solving equation (\ref{eq:qform1}) can be random but is strictly exogenous, as it depends only on the strictly exogenous part of the regressors $\tilde{X}$. Solution $\hat\Gamma$ to equation (\ref{eq:qform2}) is random and depends on error term $\varepsilon$. However, as stated in part (ii) of Theorem \ref{Thm: IV} the feasible estimator $r'\hat\beta^\text{IV}(\hat\Gamma)$ is consistent and has the same asymptotic distribution as the infeasible one using the ideal $\Gamma_0$.

An appealing feature of our proposed solution is that our estimator is linear in the outcome variable. Our proposal eliminates bias by working only with the regressors. The same bias correction would work for any outcome variables that satisfy the weak exogeneity assumption with one-period feedback to regressors. Our solution does not require one to estimate the direction of feedback $\alpha $. 

Figure \ref{fig:bias} shows that in the simulations the proposed IV estimator fixes the bias of OLS with essentially no increase in the standard deviation of the estimator.

\subsection{Consistency of variance estimator}

Let us introduce $T-K_\Gamma = \tr\big[(I-\Gamma )\tilde M_\Gamma \big]$ and define an estimator for variance $ \hat \sigma^2(\Gamma)=\frac{y'(I-\Gamma)M_\Gamma y}{T-K_\Gamma}$ for some matrix $\Gamma$. Using equation \eqref{eq: equivalence1} in the Appendix, we find that the value of $\hat \sigma^2(\Gamma)$ is invariant to the value of $\beta$ and that $\hat \sigma^2(\Gamma)=\frac{\varepsilon'(I-\Gamma)M_\Gamma \varepsilon}{T-K_\Gamma}$. Part (ii) of Lemma \ref{lem: equivalence} from the Appendix implies that $\hat \sigma^2(\Gamma)$ possesses a very desirable property for a variance estimator, namely, it is non-negative for any realization of the data. If $\Gamma$ solves equation \eqref{eq:qform1} then $\hat \sigma^2(\Gamma)$ is consistent for $\sigma^2$.

\begin{thm}\label{Thm: consistency of variance}
    Suppose Assumption \ref{Ass: OLS} holds, $\Gamma$ is $\tilde X$-measurable, and $\|\Gamma\|<1-c$. Then
    \begin{align}
        \frac{\hat \sigma^2(\Gamma)}{\sigma^2} = 1 - \frac{\sigma^2 \alpha' \bar{S}_\Gamma^{-1}\alpha }{T-K_\Gamma}\tr\!\big[D'(I-\Gamma)\tilde M_\Gamma]\tr\!\big[D(I-\Gamma)\tilde M_\Gamma] + o_p(1).
    \end{align}
    Specifically, it follows that $\hat \sigma^2(\Gamma)$ is consistent for $\sigma^2$ when $\Gamma$ is such that equation \eqref{eq:qform1} holds.
\end{thm}

As with Theorem \ref{Thm: IV}, the theoretically desirable $\Gamma$ that solves \eqref{eq:qform1} is not known, but one can search for an empirically feasible value $\hat\Gamma$ that solves \eqref{eq:qform2}. One can easily use Part (ii) of Theorem 3 to extend the argument of Theorem \ref{Thm: consistency of variance} and also show that $\hat \sigma^2(\hat\Gamma)$ is consistent.

\section{Inference}
\label{sec: inference}

In this section we show that under some additional assumptions the proposed IV estimator is asymptotically Gaussian conditionally on $\tilde X$. We also suggest standard errors that can be paired with our estimator in order to achieve asymptotically valid inferences, that is, to construct confidence sets and/or to produce reliable t-tests.

There are two theoretical and practical challenges we encounter. The first is a need to correctly account for the asymptotic importance of a quadratic form. Specifically, the bias of OLS arises due to presence of a quadratic form in errors, which has a non-trivial mean, as was noticed in \cite{sawa1978exact}. We construct our estimator to guarantee a zero mean of its corresponding quadratic form and thus eliminate the bias asymptotically. The quadratic form in error terms with zero mean is asymptotically Gaussian when its rank is growing to infinity under a condition of eigenvalue negligibility. We refer the reader to \cite{de1987central,chao2012asymptotic,anatolyev2019many,solvsten2020robust,kline2020leave} for examples of the Central Limit Theorems for quadratic forms. The naive standard errors tend to improperly account for asymptotic uncertainty of a quadratic form. The asymptotic importance of such quadratic forms has appeared previously in the literature on linear models with many instruments and/or many regressors. For example, \cite{hansen2008estimation} show the importance to adjust standard errors for presence of a quadratic form in many instrument settings. See also \cite{anatolyev2019many} for a comprehensive survey of the issue. The second challenge is that the quadratic form we end up with has non-zero diagonal elements that are on average zero. This makes the asymptotic variance depend on skewness and kurtosis of errors that are hard to estimate. This is similar to the issue of proper inference for LIML-type estimators with many instruments (as in \cite{hansen2008estimation}). Below we consider two instances when we can handle the first challenge with relative ease and ignore the second challenge.

\subsection[Inference with moderately many regressors]{Inference when $K/T\to 0$}

\begin{thm}\label{thm: inference for smaller K}
    Suppose Assumption \ref{Ass: OLS} holds, 
    $\max_t \norm*{(\tilde X'\tilde X)^{-1/2}\tilde x_t} = o_p(1)$. Then, as $T\to\infty$,
    \begin{align}
        \frac{r'\hat\beta^\emph{IV}(\hat \Gamma)-r'\beta}{\sqrt{\hat\varSigma_T}}\Rightarrow N(0,1)
    \end{align}
    where $\hat\varSigma_T=\hat\sigma^2(\hat \Gamma)\|r'(X'(I-\hat \Gamma)X)^{-1}X'(I-\hat \Gamma)\|^2.$
\end{thm}

Condition $\max_t \norm*{(\tilde X'\tilde X)^{-1/2}\tilde x_t} = o_p(1)$ is a relatively standard negligibility condition often invoked in order to obtain asymptotic Gaussianity of OLS using Lindeberg CLT \citep[see, e.g.,][page 334]{koenker1999gmm}. It implies among other things that the maximal diagonal element $\tilde P_{tt}$ of the projection matrix $\tilde P$ is asymptotically negligible. At the same time, the average of these diagonal elements is equal to $K/T$, and thus, this condition can hold only when we have a moderately large number of regressors ($K/T\to 0$). When the number of regressors is moderately large, both challenges described at the beginning of this section are asymptotically negligible. The standard errors suggested by Theorem \ref{thm: inference for smaller K} look like the usual IV-type standard errors with one important change, they use the newly proposed estimator of variance $\hat\sigma^2(\hat\Gamma)$.

\subsection{Inference with Gaussian errors}

\begin{thm}\label{thm: inference for Gaussian errors}
    Suppose Assumption \ref{Ass: OLS} holds and $\varepsilon_1$ is Gaussian conditionally on $\tilde X$. Assume that $\Gamma$ solves equation \eqref{eq:qform1} and $\norm{\Gamma} < 1-c$. 
    Then, as $T\to\infty$,
    \begin{align}
        \frac{r'\hat\beta^\emph{IV}(\Gamma)-r'\beta}{\sqrt{\varSigma_T}}\Rightarrow N(0,1)
    \end{align}
    where $\varSigma_T$ is measurable with respect to $\tilde X$. With probability asymptotically approaching one, $\varSigma_T\leq (1+\psi)\hat\sigma^2( \Gamma)\|r'(X'(I- \Gamma)X)^{-1}X'(I- \Gamma)\|^2,$ where $\psi=\frac{|\tr(B^2)|}{\tr(B'B)}$ and $ B=D'(I-\Gamma) \tilde M_\Gamma.$
\end{thm}

Theorem \ref{thm: inference for Gaussian errors} allows the number of regressors to grow proportionally with the sample size. The assumption that errors are Gaussian is used in several ways. First, it allows us to avoid imposing the Lindeberg-type negligibility condition, as any linear combination of errors with weights depending on $\tilde X$ is conditionally Gaussian in finite-samples. Secondly, it removes the need to estimate skewness and kurtosis. 

However, in the case when $K$ grows proportionally to the sample size, the standard errors stated in Theorem \ref{thm: inference for smaller K} do not correctly account for the uncertainty induced by the quadratic form. The missing term in the asymptotic variance formula depends on the importance of the feedback and is hard to estimate. We instead provide an upper bound on the asymptotic variance that makes the confidence sets asymptotically valid but conservative. In our simulation we noticed that the correction $\psi$ tends to be tiny and does not change the standard errors much. We advise the calculation of it as a robustness check.

\section{Extension to multiple periods}
\label{sec: multi period}

In the previous sections we maintained the assumption that the violation of strict exogeneity happens for one period only. Some results can be extended to feedback lasting a fixed finite number of periods.

\begin{ass}\label{ass: dgp}
    \begin{enumerate}[label=(\roman*)]
                
        \item The observed regressors $x_t$ can be decomposed as $x_t = \tilde x_t + \sum_{\ell=1}^L \alpha_\ell \varepsilon_{t-\ell}$ where the $T\times K$ matrix $\tilde{X} = [\tilde x_1,\dots,\tilde x_T]'$ has full rank.
        
        \item The errors $\{ \varepsilon_t \}_{t=1-L}^T$ are i.i.d. conditional on $\tilde{X}$ with $\E[\varepsilon_t|\tilde X]=0$, $c<\sigma^2 := \E[\varepsilon_t^2|\tilde X] <C$, and $\E[\varepsilon_t^4|\tilde X] < C$ almost surely.
        
        \item The non-random vectors $\alpha_1,\dots,\alpha_L, r \in \R^K$ satisfy ${\alpha_\ell'(\frac{1}{T}\tilde X'\tilde X)^{-1}\alpha_\ell}=O_p(1)$ for $\ell=1,\dots,L$ and $r'(\frac{1}{T}\tilde X'\tilde X)^{-1}r=O_p(1)$.
        
        \item $L$ is fixed and $ \frac{K}{T} < 1-c$.
       
    \end{enumerate}
\end{ass}

The vector $\alpha_\ell$ describes how the shock to the outcome variable affects the regressors $\ell$ periods later. The direction of the feedback may vary freely with the lag as the regressors differ in the speed of reaction/adjustment. However, we probably should expect that the size of the $\ell$-th period feedback measured as ${\alpha_\ell'(\frac{1}{T}\tilde X'\tilde X)^{-1}\alpha_\ell}$ should become negligible for large enough $\ell$ in typical macroeconomic settings.

\begin{thm}\label{thm: consistency of multi-period version}
    Suppose Assumption \ref{ass: dgp} holds, $\Gamma$ is $\tilde X$-measurable, and $\|\Gamma\|<1-c$. Then,
    \begin{align} \label{eq: some bias}
        r'\hat\beta^\emph{IV}(\Gamma)-r'\beta = \sigma^2 \sum_{\ell=1}^L r'\bar S_\Gamma^{-1}\alpha_\ell \tr\!\big[ (D')^\ell(I-\Gamma)\tilde M_\Gamma \big] + o_p(1),
      \end{align}
    where $\bar S_\Gamma = \tilde X'(I-\Gamma )\tilde X+\sigma^2\sum_{j,\ell=1}^L \alpha_j\alpha_\ell' \tr\!\big[ (D')^j(I-\Gamma )\tilde M_\Gamma D^\ell \big]$.
\end{thm}

Theorem \ref{thm: consistency of multi-period version} is a direct generalization of Part (i) of Theorem \ref{Thm: IV}. The special case of $\Gamma=0$ shows that the bias of OLS is a linear combination of $L$ terms involving the lower diagonal traces of the projection matrix $\tilde M$. Lower diagonal traces, $\tr[(D')^\ell \tilde M]$, correspond to average measures of the regressors' auto-correlations of order $\ell$, and are expected to decay for stationary regressors. Combined with the expected decrease in the size of $\alpha_\ell$, the first few terms in the bias formula should capture most of the bias in stationary applications.

The bias formula \eqref{eq: some bias} is derived for any IV-type estimator for an $\tilde X$-measurable $T\times T$ matrix $\Gamma$ with $\|\Gamma\|<1-c$. Theorem \ref{thm: consistency of multi-period version} suggests that if $\Gamma$ solves a system of $L$ equations $\tr\big[(D')^\ell(I-\Gamma)\tilde M_\Gamma \big]=0$ for $\ell=1,...,L$ then $r'\hat\beta^\text{IV}(\Gamma)$ is a consistent estimator. A natural suggestion is to search for $\Gamma$ in the class of matrices $\Gamma=\sum_{\ell=1}^L\gamma_\ell D^\ell$ where the $L$ parameters $\{ \gamma_\ell \}_\ell$ solve the system of equations. We leave the questions of considering other classes of matrices as well as establishing guarantees for existence of a solution to future research.\looseness=-1

Theorem \ref{thm: consistency of multi-period version} can be generalized further to the case of infinite feedback (under a summability condition on the feedback magnitudes). This generalization can be achieved through the use of an argument that approximates a model with infinite lags using a model with a finite but slowly growing number of lags. Specifically, Theorem \ref{thm: consistency of multi-period version} generalizes to a model where a lagged dependent outcome enters as the regressor (which yields a model with infinite feedback and geometrically decaying feedback magnitudes).\footnote{These results are available from the authors upon request, but are not included here as they would constitute a separate paper.}

The formula for the OLS  bias in time series has been derived before in special cases. Specifically, \cite{sawa1978exact} and \cite{nankervis1988exact} derived the exact finite-sample bias in an AR(1) model with normal errors and argued that the bias is sizable when the sample size is double-digit. \cite{kiviet1999alternative} considered a model with Gaussian innovations that has a small number of strictly exogenous regressors and derived a formula for asymptotic ($1/T$) bias from using the lagged outcome as a regressor. They argued that the bias can be significant in small samples. \cite{stambaugh1999predictive} showed that the weak exogeneity bias is large when a single weakly exogenous regressor in persistent. All these results are special cases of our formula when applied to infinite geometrically decaying feedback.   

\section{Simulations}
\label{sec: mcs}

The goal of this section is to assess the size of the OLS bias  in a `typical' regression using US macroeconomic data. We use the data set from \cite{stock2016dynamic} containing quarterly observations from 1964 to 2013 ($T=200$) on 108 US macro indicators. This data set is largely  similar to the \cite{mccracken2021fred} FRED-QD data set. The data set includes a  broad class of variables with diverse time series properties.

Many macro and financial indicators tend to be very persistent and may be integrated up to the second order, that is, be stationary, I(1), or I(2) processes. A prevailing (but not uniformly accepted) practice is  to transform all variables to stationary before running a regression in order to avoid issues of co-integration and near co-integration or biases related to persistent  regressors (see \cite{stambaugh1999predictive} on such biases). The way applied researchers transform variables to stationary or make decisions about variable stationarity varies widely across the literature with many  such decisions based both on statistical tests and expert judgements. Given that a large fraction of regressions is aimed at business-cycle parameters and in order to not take a stand and to unify the pre-treatment of variables we apply \cite{hamilton2018you} transformation to all variables in the data set. Specifically, for each variable we define its cyclical component to be a two-year-ahead forecast error to this variable based on a univariate AR(4) regression. According to \cite{hamilton2018you}, this filtering transforms many types of stationary and up to second order integrated variables into stationary ones and extracts their business cycle component.

For each $K$ in $\{5, 15, 25, \dots, 85, 100\}$ we perform 100 experiments where we randomly draw $K$ distinct variables from the  transformed data set, denote them $X_\text{r}$ and an additional variable $y_\text{r}$. We calculate through simulations, the biases and standard deviations of the OLS and of our proposed estimator (referred to as IV) for the linear contrast in the feedback direction in the  regression of $y_\text{r}$ on $X_\text{r}$ under the assumption of one-period violation of strict exogeneity. For this, we simulate $N=1000$ samples from a data generating process satisfying Assumption \ref{Ass: OLS} that preserves the time series behavior of regressors and the feedback size/direction of the observed $(y_\text{r},X_\text{r})$. Specifically, we use as true parameters the empirical OLS values $\beta = (X_\text{r}'X_\text{r})\inverse X_\text{r}' y_\text{r}$, $\sigma^2 = e'e/(T-K)$ for $e = y_\text{r} -X_\text{r}\beta$, and $\alpha = X_\text{r}' D' e/(e'e)$. We simulate samples as $X = X_\text{r} + D'\varepsilon \alpha'$ and $y = X \beta + \varepsilon$ where $\varepsilon \sim N(0,\sigma^2 I)$.  We calculate the bias and standard deviation for the OLS and IV estimators of the linear contrast with $\theta=\alpha'\beta$.\looseness=-1

The left panel of Figure \ref{fig:biasFRED} depicts the results of the experiments (for different $K$) that fall into the 10th percentile of the OLS bias. For those experiments, we report the OLS bias and standard deviation and the IV bias and standard deviation in the feedback direction alone along with the normalized lower trace of $M_\text{r}=I-X_\text{r}(X_\text{r}'X_\text{r})^{-1}X_\text{r}'$, that is $\tr(D'M_\text{r})/T$. The right panel contains the results of the experiments that fall into the 90th percentile of the bias.\looseness=-1

\begin{figure}[htbp]
\centering
\includegraphics[width=\columnwidth]{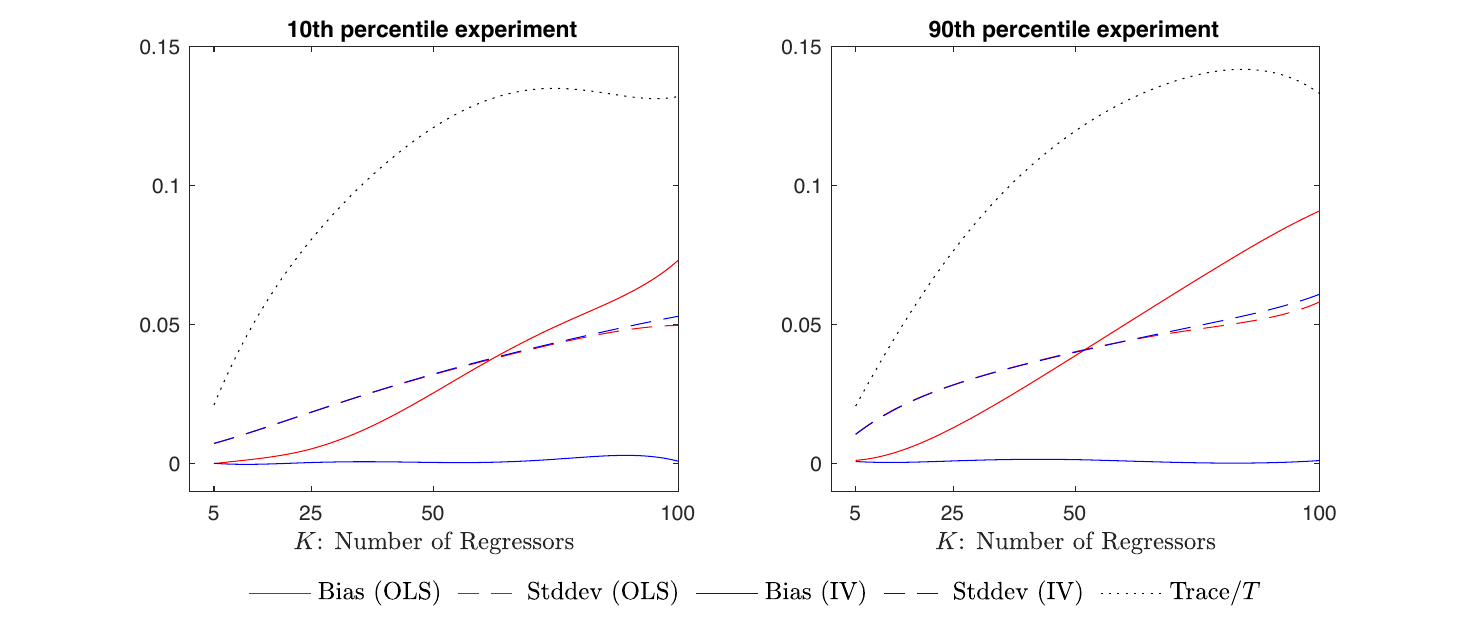}
\caption{Absolute Bias and Standard Deviation of OLS and IV with $T=200$}
\label{fig:biasFRED}
\end{figure}

Figure \ref{fig:biasFRED} shows that a typical macroeconomic data set demonstrates enough time series dependence or  short-term linear predictability that it creates a potential for substantial OLS biases. A typical size/direction of one-period feedback  for a randomly picked-up regression using macro indicators is such that for a sample with 200 time periods in a regression with 25 regressors  the OLS bias in the feedback direction is about half of the standard deviation and approximately equal to the standard deviation with 50 regressors. This leads to invalid statistical inferences when relying on the OLS. Figure \ref{fig:sizeFRED} reports the size distortions in the experiments described above for the 5\% tests about the linear contrast in the direction of the feedback. Our new proposed estimator completely corrects the bias without any significant change to the standard deviation and restores the correct size for statistical tests.\looseness=-1

\begin{figure}[htbp]
\centering
\includegraphics[width=\columnwidth]{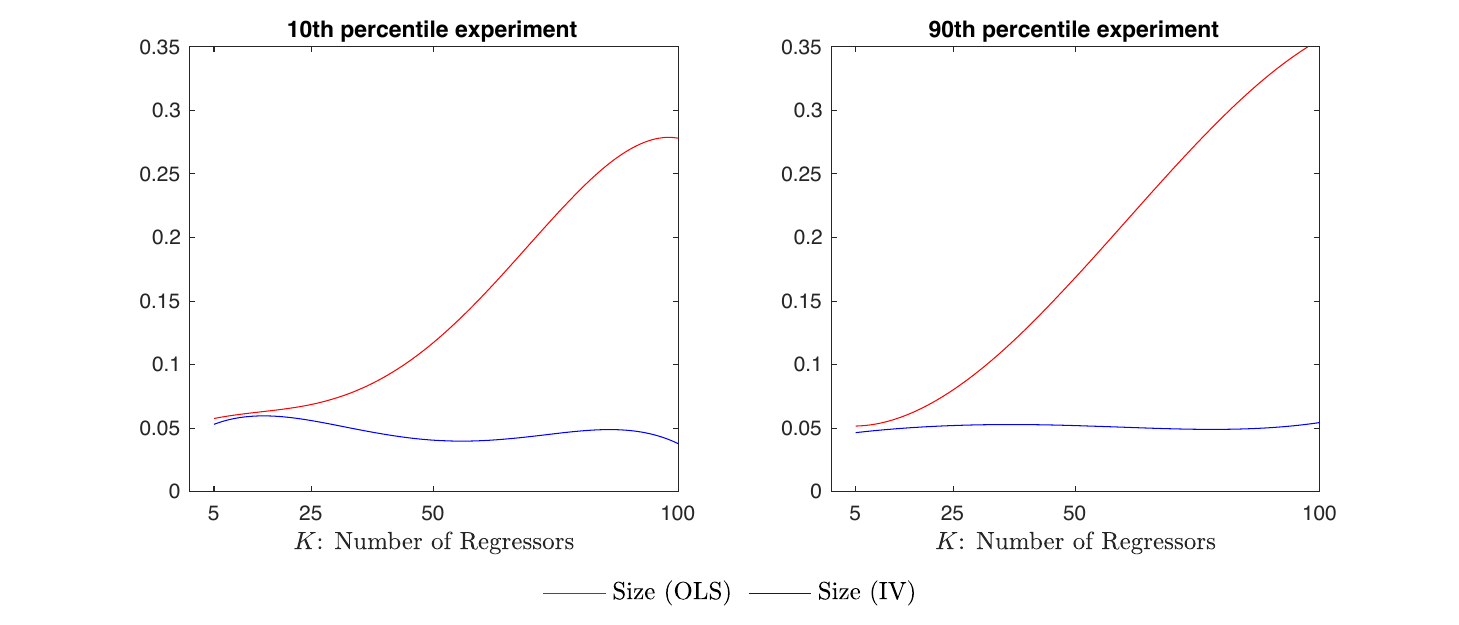}
\caption{Size of Nominal 5\% two-sided tests using OLS and IV with $T=200$}
\label{fig:sizeFRED}
\end{figure}

Another observation from Figure \ref{fig:biasFRED} is that the ratio of the lower trace of the regressor projection matrix to the sample size is highly indicative of the size of the worst bias. Applied researchers should be worried when this indicator exceeds 5--10\%. 
It is worth pointing out that in all of our experiments and simulations we encountered no problems of finding the solution to equation \eqref{eq:qform2}, which supports our assertion that the sample size requirement imposed in Lemma \ref{lem: solution exists} is sufficient but not necessary for the existence of the solution.

Finally, we note that the results presented in Figure \ref{fig:biasFRED} are both qualitatively and quantitatively similar to the left panels of Figures \ref{fig:bias} and  \ref{fig:biasMA}.  The difference between figures is that in Figures \ref{fig:bias} and  \ref{fig:biasMA}  we simulated artificial regressors following a simple AR(1) or MA(1) process. Our theoretical results are quite agnostic about the time series properties of the regressors, specifically, we make no assumptions about stationarity or origin of the regressors. The data generating processes underlying Figure \ref{fig:biasFRED} mimics the time series behavior and feedback size/direction of a 'typical' macroeconomic application to the best of our ability.

\begin{figure}[htbp]
\centering
\includegraphics[width=\columnwidth]{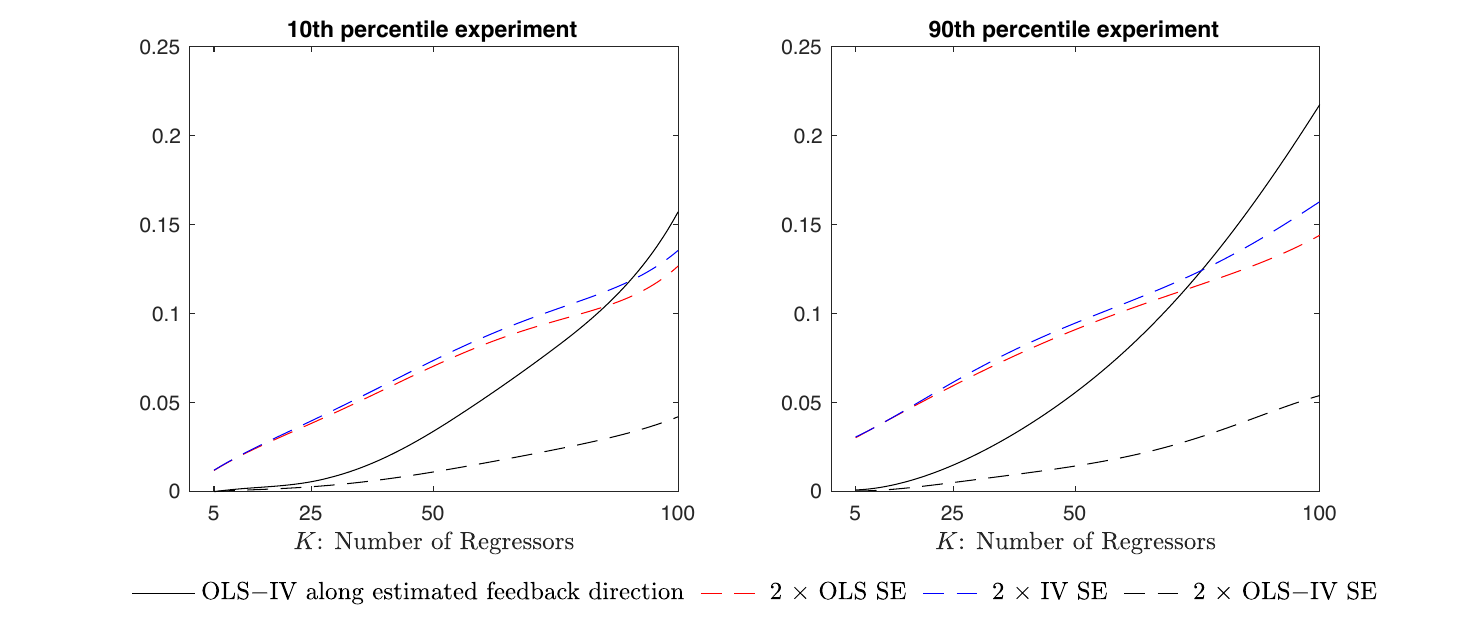}
\caption{Difference between IV and OLS along estimated feedback direction, $T=200$}
\label{fig:difFRED}
\end{figure}

Figure \ref{fig:difFRED} answers the question of how different the results of the OLS and our new proposed estimators are in a randomly selected regression based on typical macroeconomic data. As in the experiments described above we randomly select $(y_\text{r},X_\text{r})$ but rather than evaluate the theoretical bias in simulations we calculate the realization  of the difference between two estimators for the contrast in the estimated feedback direction. We report the experiment corresponding to 10th percentile  of the absolute difference on the left panel and 90th percentile on the right. For the described experiments we report the absolute value of the difference between the OLS and the IV estimators, double of the OLS and IV standard errors, and finally double the standard deviations of the difference.\footnote{One can derive the distribution of the difference under assumption of strict exogeneity using ideas analogous to those stated in Section \ref{sec: inference}.} We report doubled standard errors to directly relate the results to the corresponding $t$-statistics.

As we see from Figure \ref{fig:difFRED}, the difference between the OLS and IV estimators is statistically significant in  all experiments we report. Given that the validity of OLS with many regressors is only known for strictly exogenous case, which is a special case of Assumption \ref{Ass: OLS} with $\alpha=0$,  the observed difference between estimators are mostly due to the biases and cannot be explained by randomness in realizations. An alternative explanation is that  a hypothesis of no feedback is rejected in almost all regressions we considered. This is a sign of prevalence of feedback mechanisms in most macroeconomic applications. Comparing the difference between the two estimators to the doubled standard errors shows that when the number of regressors is above 70 the IV estimate falls outside of the standard OLS confidence set and the OLS estimate falls outside the IV confidence set. This demonstrates the potential for drastic disagreement between the two estimation methods. Finally, we notice that the standard deviation of the OLS versus IV difference is much smaller than the standard error of either the OLS or the IV estimators, in some cases more than five times smaller. This shows that stochastic deviations in two estimators are greatly aligned, and that most of the difference between the two estimators comes from the bias, not the variance.

\begin{table}[htbp]
    \centering
    {\small
        \caption{Statistical significance of the differences in OLS and IV coefficients}
        \label{tab:...}
        \begin{threeparttable}
            \begin{widetable}{.98\columnwidth}{lr rrr rrr rrr r}
                \toprule
                K && 5 & 15 & 25 & 35 & 45 & 55 & 65 & 75 & 85 & 100 \\
                \midrule
                $100*\text{ave}(t_{\Delta} > 1.96)$ && 21.20& 13.93& 11.72& 11.03& 10.47&  8.93&  8.12&  7.64&  6.91&  6.41 \\
                \bottomrule
            \end{widetable}
            \begin{tablenotes}
                \footnotesize
               \item \textsc{NOTE}: $t_{\Delta} = \abs{\hat \beta^\text{LS} - \hat \beta^\text{IV}}/se(\hat \beta^\text{LS}- \hat \beta^\text{IV})$. The average is taken over all coefficients in 100 randomly chosen models for each value of $K$.
            \end{tablenotes}
        \end{threeparttable}
    }
\end{table}

While Figure \ref{fig:difFRED} reports the difference between the two estimators in the most affected direction, one may also ask how different coefficients on individual regressors are. The average bias of individual coefficients is a counter-play of  two forces.  On one side, a larger number of regressors leads to a larger lower diagonal trace $\tr(D'M_\text{r})/T$ and thus to a larger bias in the most affected direction. At the same time, when the dimensionality of regressors is large, (a randomly selected) feedback direction is on average less aligned with any coordinate direction. Thus, the same size of the worst bias results in a smaller average individual coefficient bias when the number of regressors increases, as it spreads out among many individual coefficients.\looseness=-1

In Table \ref{tab:...}, we report the average fraction of coefficients that display a statistically significant difference between the OLS and the IV estimates. The average is taken over the $K$ coefficients in each regression and over the 100 random regressions with $K$ regressors drawn from the macroeconomic data base described above. While the fraction of coefficients with a statistically significant difference between OLS and IV is declining with $K$, the absolute number of such coefficients increases. We conclude that while most directions/coefficients are immune to the biases, a non-trivial fraction of coefficients are significantly affected. \looseness=-1


\bibliographystyle{chicago}
\bibliography{literature.bib}

\begin{appendices}

\section{Proofs}
\label{sec: proofs}

\paragraph{Notation} For brevity, the expectation $\E[\cdot]$ is used to denote the conditional expectation $\E[\cdot|\tilde X]$. The proofs use some well known identities involving matrix traces. For matrices of compatible dimensions, $\tr(ABD)=\tr(BDA)$, $\tr(A)=\tr(A')=\tr(A+A')/2$. As in the main paper, let $P=I-M=X(X'X)^{-1}X'$ and $P_\Gamma=I-M_\Gamma=X(X'(I-\Gamma)X)^{-1}X'(I-\Gamma).$

\subsection{Auxilliary lemmas}

The results of the next lemma are well-known, but included here for an ease of reference.

\begin{lemma}\label{lem: trace}
    \begin{enumerate}[label=(\roman*)]
        \item For a symmetric matrix $A$ and positive semi-definite (psd) matrix $B$, we have the bounds $\lambda_{\min}(A) \tr(B) \leq \tr(AB) \leq \lambda_{\max}(A) \tr(B)$. \label{lem: trace 1}
        \item For any square matrix $A$, let $\lambda$ be the smallest eigenvalue of $\frac{A+A'}{2}$. If $\lambda > 0$, then $\norm{A^{-1}} \leq 1/\lambda$. \label{lem: trace 2}
        \item For any compatible matrices $A$ and $B$, we have $\abs{\tr(AB)} \le \norm{A}_F \norm{B}_F$. \label{lem: trace 3}
    \end{enumerate}
\end{lemma}

\begin{proof}[Proof of Lemma \ref{lem: trace}.]
    \begin{enumerate*}[mode=unboxed]
        \item[\ref{lem: trace 1}] As $A$ is symmetric, there exists $U$ with $U'U=I$ such that $D = UAU'$ is a diagonal matrix with eigenvalues of $A$ along its diagonal. Let $F = UBU'$. Note that $F_{ii}\geq 0$ since $B$ is psd and $\tr(B) =\tr(F) = \sum_i F_{ii}$. Now,
        \begin{align}
            \tr(AB)=\tr(UAU'UBU')=\sum\nolimits_i D_{ii} F_{ii} \leq \max_i D_{ii}\sum\nolimits_i F_{ii} = \lambda_{\max}(A) \tr(B),    
        \end{align}
        and $\sum_i D_{ii} F_{ii} \geq \min_{i} D_{ii} \sum_i F_{ii} = \lambda_{\min}(A) \tr(B)$.
    
        \item[\ref{lem: trace 2}] Now, $\lambda \norm{x}^2 \leq x'(A+A')x/2 = x'Ax \leq \norm{x}\norm{Ax}$. As $\lambda > 0$, it follows that $Ax \neq 0$ when $x \neq 0$, so $A$ is invertible. For $x \neq 0$, we then have $\lambda \leq \norm{Ax}/\norm{x} = \norm{y}/\norm{A\inverse y}$ where $y = Ax$. Thus, $\norm{A\inverse} = \sup_{y \neq 0} \norm{A\inverse y}/\norm{y} \le 1/\lambda$. 
    
        \item[\ref{lem: trace 3}] $\abs{\tr(AB)} = \abs*{\sum\nolimits_{ij} A_{ij} B_{ij}} \leq \sqrt{\sum\nolimits_{ij} A_{ij}^2} \sqrt{\sum\nolimits_{ij} B_{ij}^2} = \norm{A}_F \norm{B}_F.$ \qedhere
    \end{enumerate*}
\end{proof}

\begin{lemma}\label{lem: equivalence}
    Suppose $X \in \R^{T \times K}$ has full rank and $\Gamma\in \R^{T \times T}$ has $\|\Gamma\|<1-c$. Then
    \begin{enumerate*}[label=(\roman*), mode=unboxed, itemjoin={{}}]
        \item the matrices $I-\Gamma$, $I - P\Gamma$, $X'(I-\Gamma)X$ and $I+A_\Gamma M$ are invertible where $A_\Gamma=(I-\Gamma)\inverse \Gamma$; \label{lem: equivalence inv}

        \item the following identities hold: \label{lem: equivalence 2}
	\begin{align}
	     (I-\Gamma) M_{\Gamma} &= M(I+A_\Gamma M)\inverse , \label{eq: equivalence1}\\
	   P_\Gamma &= (I-P\Gamma)\inverse P(I-\Gamma), \label{eq: equivalence4}\\
	   M_\Gamma &= (I-P\Gamma)\inverse M, \label{eq: equivalence2}\\
          \hat\beta^\emph{IV}(\Gamma) &=(X'X)^{-1}X'(I+A_\Gamma M)^{-1}y. \label{eq: equivalence3}
	\end{align}

        \item the matrix $(I-\Gamma)M_\Gamma+M_{\Gamma}'(I-\Gamma')$ is positive semi-definite and $c<\frac{T-K_\Gamma}{T-K}<C$ where $T-K_\Gamma=\tr[ (I-\Gamma) M_{\Gamma}]$; \label{lem: equivalence 1}
    \end{enumerate*}
\end{lemma}

\begin{proof}[Proof of Lemma \ref{lem: equivalence}.]
    \begin{enumerate*}[mode=unboxed, itemjoin={{}}]
    
        \item[\ref{lem: equivalence inv}] Lemma \ref{lem: trace}\ref{lem: trace 2}, the triangle inequality, and $\norm{\Gamma} < 1$ yields invertibility of $I-\Gamma$, $I-P\Gamma$, $X'(I-\Gamma)X$, and $I-\Gamma P$. $I + A_\Gamma M$ is invertible since
        \begin{align}\label{eq: relation between AM and P}
            I+A_\Gamma M=(I-\Gamma)^{-1}(I-\Gamma+\Gamma M)=(I-\Gamma)^{-1}(I-\Gamma P). 
        \end{align}
        
        \item[\ref{lem: equivalence 2}] Next, we note that
        \begin{align}
            \big(X'X\big)\inverse X' &= \big(X'X\big)\inverse \big[ X'(I-\Gamma)X \big(X'(I-\Gamma)X\big)\inverse \big] X' \\
            &= \big(X'X\big)\inverse X'(I-\Gamma)P_\Gamma (I-\Gamma)\inverse. \label{eq: Sgamma} 
        \end{align}
        Pre-multiplying \eqref{eq: Sgamma} by $X$ and using $P_\Gamma = P P_\Gamma$ gives us $P=(I-P\Gamma)P_\Gamma(I-\Gamma)^{-1}$ and hence \eqref{eq: equivalence4}. Therefore, we also have \eqref{eq: equivalence2}:
        \begin{align}
            M_\Gamma=I-P_\Gamma=(I-P\Gamma)^{-1}(I-P\Gamma-P(I-\Gamma))=(I-P\Gamma)^{-1}M.
        \end{align}
        Using the ``push-through'' identity $(I-P\Gamma)\inverse P= P(I-\Gamma P)\inverse $ on \eqref{eq: equivalence4} similarly yields \eqref{eq: equivalence1}:
        \begin{align}
            M_\Gamma &= I-P(I+A_\Gamma M)\inverse = (I+A_\Gamma)M(I+A_\Gamma M)\inverse = (I-\Gamma)\inverse M(I+A_\Gamma M)\inverse. 
        \end{align}
        Reusing that $P_\Gamma = P(I+A_\Gamma M)\inverse$ we get \eqref{eq: equivalence3} from $\hat\beta^\text{IV}(\Gamma)=\big(X'X\big)\inverse X'P_\Gamma y$ and $XP=X$. \
    
        \item[\ref{lem: equivalence 1}] Positive semi-definiteness comes from \eqref{eq: equivalence1}: 
        \begin{align}
            (I-\Gamma)M_\Gamma+M_{\Gamma}'(I-\Gamma')= (I+MA_\Gamma')^{-1}M\left\{ (I-\Gamma')\inverse + (I-\Gamma)\inverse \right\} M(I+A_\Gamma M)^{-1} 
        \end{align}
        and observing that $(I-\Gamma')^{-1}+(I-\Gamma)^{-1}$ is psd. The rate condition follows from $\tr(M)=T-K$, Lemma \ref{lem: trace}\ref{lem: trace 1},
        \begin{align}
            T-K_\Gamma=&\tr[ M(I+A_\Gamma M)\inverse] =\tr\!\left[\tfrac{(I+A_\Gamma M )\inverse+(I+MA_\Gamma')\inverse }{2}M\right]\!,
        \end{align}
        and that the eigenvalues of $\frac{(I+A_\Gamma M )\inverse+(I+MA_\Gamma')\inverse }{2}$ are in $\frac{1-\norm{\Gamma}}{1+\norm{\Gamma}}$ to $\frac{1+\norm{\Gamma}}{1-\norm{\Gamma}}$ with $\norm{\Gamma} < 1-c$.\qedhere
    \end{enumerate*}
\end{proof}

\begin{lemma}\label{lem: rotation}
    Suppose $y_t = x_t'\beta + \varepsilon_t$, $x_t = \tilde x_t + \sum_{\ell=1}^L \alpha_\ell \varepsilon_{t-\ell}$ where the $T\times K$ matrix $\tilde{X} = [\tilde x_1,\dots,\tilde x_T]'$ has full rank, and $\alpha_1, \dots, \alpha_L, r$ are $K\times 1$ vectors. Then, there exists an invertible $K\times K$ matrix $\Theta$ mapping $\{\tilde X,r,\beta,\{\alpha_l\}_{l=1}^L \}$ to $\{\tilde X\Theta,\Theta'r,\Theta\inverse\beta,\{\Theta'\alpha_l\}_{l=1}^L \}$, that satisfies: 
    \begin{enumerate}[label=(\roman*)]
        \item $\Theta'r$ and $\{\Theta'\alpha_\ell\}_{\ell=1}^L$ are spanned by the first $L+1$ basis vectors so that $r'\Theta = (r'_*,\mathbf{0}_{K-L-1}')$ and $\alpha_\ell'\Theta = (\alpha'_{*,\ell},\mathbf{0}_{K-L-1}')$ with $r_*,\alpha_{*,1},\dots,\alpha_{*,L} \in \R^{L+1}$; \label{lem: rotation 1}
        
        \item $\Theta'\tilde X_{1}'\tilde X_1\Theta/T=I_{L+1}$ and $\Theta'\tilde Z_2'\tilde X_1\Theta=0$ for $\tilde X = [\tilde X_1, \tilde X_2]$, and $\tilde Z=[\tilde Z_1, \tilde Z_2] = (I-\Gamma')\tilde X\Theta$, where $\tilde X_1$ and $ \tilde Z_1$ each have $L+1$ columns. \label{lem: rotation 2}
    \end{enumerate}
\end{lemma}

\begin{proof}[Proof of Lemma \ref{lem: rotation}.] 
    Let $\Theta = \Theta_0\inverse \Theta_1$, where $K\times K$ matrix $\Theta_0$ is the symmetric square root of $[\tilde X_1,\tilde Z_{2}]'[\tilde X_1,\tilde Z_{2}]/T$ and $\Theta_1 R_1$ is the QR decomposition of $\Theta_0^{-1}[r,\alpha_1,\dots,\alpha_L]$. Specifically, $\Theta_1' = \Theta_1\inverse$ is a $K \times K$ matrix and $R_1$ is spanned by the first $L+1$ basis vectors. We then have that $\Theta'[r,\alpha_1,\dots,\alpha_L] = \Theta_1'\Theta_0\inverse [r,\alpha_1,\dots,\alpha_L] = \Theta_1' \Theta_1 R_1 = R_1$ while we also have $\Theta'[\tilde X_1,\tilde Z_{2}]'[\tilde X_1,\tilde Z_{2}]\Theta/T = \Theta_1' \Theta_0\inverse \Theta_0^2 \Theta_0\inverse \Theta_1 = I_K$. 
    \qedhere
\end{proof}

\begin{lemma}\label{lem: Frisch-Waugh}
    Suppose $X$ and $Z$ are $T\times K$ matrices with $Z'X$ invertible. Let $X=[X_1,X_2]$ and $Z=[Z_1,Z_2]$, where $X_\ell$ and $Z_\ell$ are $T\times K_\ell$ with $K_1+K_2=K$. Define the oblique projections $P_Z = X(Z'X)^{-1}Z' = I - M_Z$ and $P_2 = X_2(Z_2'X_2)^{-1}Z_2' = I-M_2$.
    \begin{enumerate}[label=(\roman*)]
        \item (Generalized Frisch-Waugh-Lowell) If $r=[r_1'; \mathbf{0}_{K_2}]'$ with $r_1\in \R^{K_1}$, then \label{lem: Frisch-Waugh 1}
        \begin{align}
            r'(Z'X)^{-1}Z'=r_1'(Z_1'M_2X_1)^{-1}Z_1'M_2.
        \end{align}
        
        \item $M_Z=M_2-M_2X_1(Z_1'M_2X_1)^{-1}Z_1'M_2.$ \label{lem: Frisch-Waugh 2}
    
        \item If $Z=(I-\Gamma')X$ for $T\times T$ matrix $\Gamma$ with $\|\Gamma\| < 1-c$ and $A$ is a $T\times T$ matrix, then \label{lem: Frisch-Waugh 3}
        \begin{align}
            \left|\tr(A(M_Z-M_2))\right|\leq C K_1 \|A\|.
        \end{align}
    \end{enumerate}
\end{lemma}

\begin{proof}[Proof of Lemma \ref{lem: Frisch-Waugh}.] 
    Letting $\Delta=(Z_1'M_2X_1)^{-1}$, the matrix block inversion formula gives
    \begin{align}\label{eq: block inversion}
        (Z'X)^{-1} = \begin{bmatrix}
                \Delta & -\Delta Z_1'X_2(Z_2'X_2)^{-1} \\
                -(Z_2'X_2)^{-1}Z_2'X_1\Delta & (Z_2'X_2)^{-1}\{I+Z_2'X_1\Delta Z_1'X_2\}
            \end{bmatrix}
    \end{align}
    \begin{enumerate*}[mode=unboxed, itemjoin={{}}]
        \item[\ref{lem: Frisch-Waugh 1}] From \eqref{eq: block inversion} we have
        \begin{align}
            r'(Z'X)^{-1}Z'=r_1'\Delta Z_1'-r_1'\Delta Z_1'X_2(Z_2'X_2)^{-1} Z_2'=r_1'\Delta Z_1'M_2.
        \end{align}

        \item[\ref{lem: Frisch-Waugh 2}] Denote $\delta=X_1\Delta Z_1'$. Using \eqref{eq: block inversion} above:
        \begin{align}
            P_Z = X(Z'X)^{-1}Z' =\delta -\delta P_2-P_2\delta +P_2+P_2\delta P_2=I-M_2+M_2\delta M_2.
        \end{align}

        \item[\ref{lem: Frisch-Waugh 3}] We impose without loss of generality that $X_2'X_1=0$ and $X_1'X_1=I_{K_1}$. This entails no loss since \eqref{eq: equivalence1} and \eqref{eq: equivalence2} yields $Z_2'M_2=0$ and $M_2X_2=0$, which in turn implies that $(M_Z,M_2)$ is invariant under the transformation $[X_1,X_2] \mapsto [M^* X_1 (X_1'M^*X_1)^{-1/2},X_2]$ where $M^* = I - X_2(X_2'X_2)\inverse X_2' = I-P^*$. From \ref{lem: Frisch-Waugh 2}, Lemma \ref{lem: trace}\ref{lem: trace 3}, and $\|\Psi\|_F\leq \|\Psi\|\sqrt{\rk(\Psi)}$:
        \begin{align}
            \left|\tr(A(M_\Gamma-M_2))\right| 
            &=\left|\tr\left(Z_1'M_2A M_2X_1(Z_1'M_2X_1)^{-1}\right)\right|\\
           &\leq K_1 \|Z_1'M_2A M_2X_1\|\cdot\|(Z_1'M_2X_1)^{-1}\| \\ 
           &\leq K_1 \|I-\Gamma\|\cdot\|M_2\|^2\cdot\|A\|\cdot\|(Z_1'M_2X_1)^{-1}\|.
        \end{align}
        Equation \eqref{eq: equivalence2} gives $M_2=(I-P^*\Gamma)^{-1}M^*$ and therefore $\|M_2\|\leq\frac{1}{1-\norm{\Gamma}} < C$. Together with $X_2'X_1=0$, \eqref{eq: equivalence2} also yields
        \begin{align}
        Z_1'M_2X_1 &= X_1'(I-\Gamma)(I-P^*\Gamma)^{-1}M^*X_1=X_1'(I-\Gamma)(I-P^*\Gamma)^{-1}X_1,
        \shortintertext{and therefore}
        \tfrac{1}{2}(Z_1'M_2X_1+X_1'M_2'Z_1) & = \tfrac{1}{2} X_1'\left\{(I-\Gamma)(I-P^*\Gamma)^{-1}+(I-\Gamma'P^*)^{-1}(I-\Gamma')\right\}X_1.
        \end{align}
        As $X_1'X_1 = I_{K_2}$, the eigenvalues of the last matrix are larger than $\frac{1-\norm{\Gamma}}{1+\norm{\Gamma}}$. Therefore, Lemma \ref{lem: trace}\ref{lem: trace 2} and $\norm{\Gamma} < 1-c$ implies that $\|(Z_1'M_2X_1)^{-1}\| < C$.
        \end{enumerate*}
\end{proof}

\begin{lemma}\label{lem: order of terms infinite lag}
    Suppose $\tilde X$ is a $T \times k$ matrix with $\tilde X'\tilde X/T=I_{k}$, $A$ is a $T\times T$ matrix that is $\tilde X$-measurable, $\alpha_0,\dots,\alpha_{k}$ are non-random $k \times 1$ vectors with $\norm{\alpha_\ell} < C$ for all $\ell$, and $\{ \varepsilon_t \}_{t=1-k}^T$ are i.i.d. conditionally on $\tilde X$ with $\E[\varepsilon_t \!\mid\! \tilde X]=0$, $0 < \sigma^2=\E[\varepsilon_t^2 \!\mid\! \tilde X] < C$ and $\E[\varepsilon_t^4 \!\mid\! \tilde X]<C$. Let $u_\ell = (\varepsilon_{1-\ell},\dots,\varepsilon_{T-\ell})'$. Then, as $T\to\infty$ while $k$ is fixed:
    \begin{enumerate*}[label=(\roman*), mode=unboxed, itemjoin={{ }}]
        \item $\norm{T^{-1/2}\tilde X'A\sum_{\ell=0}^{k} u_\ell \alpha_{\ell}'}_F = O_p\left(\norm{A}\right)$; \label{lem: order of terms infinite lag 1}
        \item $\varepsilon'A\varepsilon-\E [\varepsilon'A\varepsilon\!\mid\! \tilde X] =O_p\left(\|A\|_F\right)$. \label{lem: order of terms infinite lag 2}
    \end{enumerate*}
\end{lemma}

\begin{proof}[Proof of Lemma \ref{lem: order of terms infinite lag}.]
    \begin{enumerate*}[mode=unboxed, itemjoin={{ }}]
        \item[\ref{lem: order of terms infinite lag 1}] Note that:
        \begin{align}
            \E\norm*{\frac{\tilde X'A}{\sqrt{T}}\sum_{\ell=0}^{k} u_\ell\alpha_{\ell}'}_F^2 
            &= \tr\!\left[\frac{\tilde X'\tilde X}{T}A\sum_{\ell,j=0}^{k} \E[u_\ell u_j']\alpha_{\ell}'\alpha_{j}A'\right]\! 
            \leq \tr\!\left[\frac{\tilde X'\tilde X}{T}\right]\! \norm*{A\sum_{\ell,j=0}^{k} \E[u_\ell u_j']\alpha_{\ell}'\alpha_{j}A'} \\
            &\leq k \norm{A}^2 \!\left(\sum_{\ell,j=0}^{k}|\alpha_{\ell}'\alpha_{j}|\cdot \norm*{\E[u_\ell u_j']}\right)\! 
            \leq C k^3 \norm{A}^2.
        \end{align}
        The last inequality uses that $u_\ell$ and $u_j$ has $\E[u_\ell u_j']=\sigma^2 D^{\ell-j}$ when $\ell>j$, $\E[u_\ell u_j']=\sigma^2 (D')^{j-\ell}$ when $\ell<j$, and $\E[u_\ell u_\ell']=\sigma^2 I$. In all cases, $\norm{\E[u_\ell u_j']}\leq \sigma^2$.
    
        \item[\ref{lem: order of terms infinite lag 2}] We have 
        \begin{align}
            \varepsilon'A\varepsilon-\E [\varepsilon'A\varepsilon] 
            = \sum_{t}\sum_{s\neq t}A_{ts}\varepsilon_t\varepsilon_s+\sum_tA_{tt}(\varepsilon_t^2-\sigma^2).
        \end{align}
        In the summations above summands are correlated only when $\{t,s\}=\{t',s'\}$. Therefore
        \begin{align}
            \E&\Big[ \Big( \varepsilon'A\varepsilon - \tilde\E[\varepsilon'A\varepsilon] \Big)^2 \Big]
            = \sigma^4\sum_{t}\sum_{s\neq t}(A_{ts}A_{st}+A_{ts}^2)+\sum_tA_{tt}^2 \tilde\E\big[(\varepsilon_t^2-\sigma^2)^2 \big] \\ 
            &\leq
            C\left(\sum_{t}\sum_{s\neq t}(A_{ts}A_{st}+A_{ts}^2)+\sum_tA_{tt}^2\right)\leq 2C\sum_{t,s}A_{ts}A_{st}=2C\|A\|_F^2,
        \end{align}
        for $C=\max\big\{\sigma^4, \ \E[(\varepsilon^2_t-\sigma^2)^2]\big\}$. We used that $|\sum_{t}\sum_{s\neq t}A_{ts}A_{st}|\leq\sum_t\sum_{s\neq t}A_{ts}^2$. 
    \end{enumerate*}
\end{proof}

\begin{lemma}\label{lem: conditions for CLT} 
    Suppose $B=D'(I-\Gamma)M_\Gamma$ where $\Gamma=\gamma D$, $|\gamma|< 1-c$ and $\tr(B) = O(K)$. Then,
    \begin{enumerate*}[label=(\roman*), mode=unboxed, itemjoin={{ }}]
        \item $\sum_tB_{tt}^2=O(K);$ \label{lem: conditions for CLT 1} 
        
        \item $\tr(B^2)=\sum_{t,s}B_{ts}B_{st}=O(K);$ \label{lem: conditions for CLT 2} 
        
        \item $\frac{T-K}{\tr(B'B)}=\frac{T-K}{\sum_{s,t}B_{s,t}^2}=O(1)$. \label{lem: conditions for CLT 3} 
    \end{enumerate*}
\end{lemma}

\begin{proof}[Proof of Lemma \ref{lem: conditions for CLT}.]
    Due to Lemma \ref{lem: trace}\ref{lem: trace 1} and equation \eqref{eq: equivalence4} we have, for any matrix $A$,
    \begin{align}
        |\tr(AP)|&=\left|\tr\left(\frac{A+A'}{2}P\right)\right|\leq \left\|\frac{A+A'}{2}\right\|\tr(P)\leq \|A\|K;\label{eq: trace of P 1}
        \\ \label{eq: trace of P}
        |\tr(AP_\Gamma)|&=|\tr((I-\Gamma)A(I-P\Gamma)^{-1}P)|\leq \frac{1+\|\Gamma\|}{1-\|\Gamma\|} \|A\|K.
    \end{align}
    \begin{enumerate*}[mode=unboxed, itemjoin={{}}]
        \item[\ref{lem: conditions for CLT 1}] $B=D'(I-\Gamma)-D'(I-\Gamma)P_\Gamma$, thus $B_{tt}=-\gamma-F_{tt}$, where $F=D'(I-\Gamma)P_\Gamma$. The condition $\tr(B)=O(K)$ implies $-\gamma T=\sum_t F_{tt} + O(K)$ and $\sum_t B_{tt}^2=\sum_t F_{tt}^2-T\gamma^2 + O(K) \leq \sum_t F_{tt}^2 + O(K)$. Consider diagonal elements of the matrix $F=AP_\Gamma$ with $A=D'(I-\Gamma)$ and $\|A\|\leq 1+|\gamma|$:
        \begin{align}
            \sum_t F_{tt}^2 &= \sum_t(AP_\Gamma)_{tt}^2
            =\sum_t \left(\sum_sA_{ts}P_{\Gamma,st}\right)^2
            \leq \sum_t \left\{\sum_sA_{ts}^2\sum_{s}P_{\Gamma,st}^2 \right\} \\
            &=\sum_t(AA')_{tt}(P_\Gamma P_\Gamma')_{tt} 
            \leq \|A\|^2\tr(P_\Gamma P_\Gamma')
            \leq C K.
        \end{align}
    
        \item[\ref{lem: conditions for CLT 2}] Next we have
        \begin{align} 
            \tr(B^2) &= \tr(D'(I-\Gamma)(I-P_{\Gamma})D'(I-\Gamma)(I-P_{\Gamma})) \\
            &=\tr[D'(I-\Gamma)D'(I-\Gamma)]+\tr\!\left[ D'(I-\Gamma)\left\{-2D'(I-\Gamma)+ P_\Gamma D'(I-\Gamma)\right\}P_{\Gamma}\right]\!. \quad
            \label{eq: some eq }
        \end{align}
        Since $\sum_tB_{tt}^2\geq 0$ reasoning above gives us that $$\tr(D'(I-\Gamma)D'(I-\Gamma))=T\gamma^2\leq \sum_t F_{tt}^2=O(K).$$ The second term in \eqref{eq: some eq } is $O(K)$ due to \eqref{eq: trace of P}. \
    
        \item[\ref{lem: conditions for CLT 3}] We use that for any matrix $A$, $\tr(A'A)=\sum_{t,s}A_{ts}^2\geq |\sum_{t,s}A_{ts}A_{st}|=|\tr(A^2)|:$
        \begin{align*}
            \tr(B'B) &= \tr(M_\Gamma'(I-\Gamma')(I-\Gamma)M_\Gamma)
            \geq (1-|\gamma|)^2\tr(M_\Gamma'M_\Gamma)\\ 
            &\geq (1-|\gamma|)^2\tr(M_\Gamma^2)
            =(1-|\gamma|)^2\tr(M_\Gamma)
            =(1-|\gamma|)^2(T-K). \qedhere
        \end{align*}
    \end{enumerate*}
\end{proof}

\begin{lemma}\label{lem: contraction}
Suppose that $\tilde X$ is a $T\times K$ matrix of rank $K$ and $\Gamma=\gamma D$.
\begin{enumerate}[label=(\roman*)]
    \item Equation \eqref{eq:qform1} holds if and only if $\gamma$ is a fixed point of the transformation $f$ given by \label{lem: contraction 1}
    \begin{align}\label{eq: contraction}
        f(\gamma)=\frac{\tr(D'\tilde M_{\Gamma})}{T-K}.
    \end{align}

    \item If $|\tr(D'\tilde M)|\leq \mu^2 K$ and $K < T/(1+(1+\mu)^2)$ for some $\mu \in [0,1]$, then $f$ is a contraction on $[-\mu,\mu]/(1+\mu)$ with Lipshitz constant strictly less than $\mu$. \label{lem: contraction 2}
\end{enumerate}
\end{lemma}

\begin{proof}
\begin{enumerate*}[mode=unboxed]
    \item[\ref{lem: contraction 1}] 
    Since $\Gamma=\gamma D$ we have $D'(I-\gamma D)\tilde M_\Gamma=D'\tilde M_\Gamma-D'D\gamma \tilde M_\Gamma.$ We can therefore re-write \eqref{eq:qform1} as: $\tr(D'\tilde M_\Gamma)-\gamma(T-K)=0.$ This equation is solved if (and only if) $\gamma = f(\gamma)$.

    \item[\ref{lem: contraction 2}] Equation \eqref{eq: equivalence2} yields $\tilde M_\Gamma=\tilde M+\gamma \tilde P D \tilde M_\Gamma$ and $\|\tilde M_\Gamma \|\leq \frac{1}{1-|\gamma|}$. Equation \eqref{eq: trace of P 1} gives $|\tr(D'\tilde PD\tilde M_\Gamma)|\leq K \|D\tilde M_\Gamma D'\|$. Therefore,
    \begin{align}
        &|f(\gamma)| 
        \leq \frac{|\tr(D'\tilde M)|}{T-K}+|\gamma|\frac{|\tr(D'\tilde PD\tilde M_\Gamma)|}{T-K}
        \leq \frac{\mu^2 K}{T-K}+|\gamma|\frac{K}{(T-K)(1-|\gamma|)}
        < \frac{\mu}{1+\mu} \\
        \shortintertext{and using equation \eqref{eq: equivalence3}}
        &\frac{\abs*{f(\gamma_1) - f(\gamma_2)}}{\abs*{\gamma_1 - \gamma_2}} 
        = \frac{\abs{\tr[D'(\tilde M_{\Gamma_1}-\tilde M_{\Gamma_2})]}}{\abs{\gamma_1-\gamma_2}(T-K)} 
        = \frac{\abs{\tr[D'\tilde P D (I-\tilde P\Gamma_2)\inverse \tilde M_{\Gamma_1}]}}{T-K} \\
        &\le \frac{K}{T-K} \frac{1}{1-\abs{\gamma_1}}\frac{1}{1-\abs{\gamma_2}} 
        < \mu
    \end{align}
    where the strict inequalities use $K < T/(1+(1+\mu)^2)$ and $|\gamma|,\abs{\gamma_1},\abs{\gamma_2}\le \frac{\mu}{1+\mu}$.
\end{enumerate*}
\end{proof}

\subsection{Proofs for results stated in the main text}

\begin{proof}[Proof of Theorem \ref{thm- OLS}.]
    Special case of Theorem \ref{Thm: IV}\ref{Thm: IV 1} with $\Gamma=0$.
\end{proof}

\begin{proof}[Proof of Theorem \ref{them: inconsistency of variance}.]
    Special case of Theorem \ref{Thm: consistency of variance} with $\Gamma=0$.
\end{proof}

\begin{proof}[Proof of Lemma \ref{lem: solution exists}.]
    \begin{enumerate*}
        \item[\ref{lem: solution exists 1}] Special case of \ref{lem: solution exists 2} with $\mu=1$ since $\abs{\tr(D'\tilde M)} = \abs{\tr(D'\tilde P)} \le K$ follows from \eqref{eq: trace of P 1}.

        \item[\ref{lem: solution exists 2}] Follows from Lemma \ref{lem: contraction} and the Banach fixed point theorem.
    \end{enumerate*}
\end{proof}

\begin{proof}[Proof of Theorem \ref{Thm: IV}.]
\begin{enumerate*}[mode=unboxed]
    \item[\ref{Thm: IV 1}] Special case of Theorem \ref{thm: consistency of multi-period version} with $L=1$.

    \item[\ref{Thm: IV 2}] Let $f(\gamma)$ be as in equation \eqref{eq: contraction} define its empirical analog $\hat f(\gamma)=\frac{tr(D'M)}{T-K}+\gamma\frac{tr(D'PDM_\Gamma)}{T-K}$ where $\Gamma=\gamma D$. Since $K < T/5$, Lemma \ref{lem: contraction}\ref{lem: contraction 2} yields that both $f$ and $\hat f$ are contractions on $[-1/2,1/2]$ with contraction speed bounded by $\frac{1}{2}$ and therefore has unique fixed points $\gamma_0$ and $\hat \gamma$ by the Banach fixed point theorem. 
    
    Furthermore, $\abs{\hat f(\hat\gamma)-\hat f(\gamma_0)} \le \frac{1}{2} \abs{\hat\gamma-\gamma_0}$. For any $\gamma$ we have:
    \begin{align}
        \abs*{\hat f(\gamma)-f(\gamma)} \le 
        \frac{1}{T-K}\left(\abs*{\tr[ D'(P-\tilde{P})] +\gamma \tr\!\left[D'PDM_\Gamma-D'\tilde PD\tilde M_\Gamma\right]}\right).
    \end{align}
    Consider the transformation $\Theta$ of Lemma \ref{lem: rotation}. Since the projections $P,\tilde P, M_\Gamma,\tilde M_\Gamma$ are invariant to linear transformations we may assume that Lemma \ref{lem: rotation}(\ref{lem: rotation 1} and \ref{lem: rotation 2}) hold with $L=1$ and $\Theta=I_K$. This implies that $M_2=\tilde M_2$, where $M_2$ is defined as in Lemma \ref{lem: Frisch-Waugh} using $X$, $(I-\Gamma')X$, and $K_1=2$, while $\tilde M_2$ is an analogously defined starting from $\tilde X$, $(I-\Gamma')\tilde X$, and $K_1=2$. Therefore, Lemma \ref{lem: Frisch-Waugh}\ref{lem: Frisch-Waugh 3} implies for any compatible matrix $A$ that
    \begin{align}
        \abs*{\tr\!\big[ A(M_\Gamma-\tilde M_\Gamma)\big]\!} 
        \le \abs*{\tr\!\big[ A(M_\Gamma- M_2)\big]\!} 
        + \abs*{\tr\!\big[ A(\tilde M_\Gamma-M_2)\big]\!}
        \leq C \|A\|
    \end{align}
    A similar statement holds with $(P,\tilde P)$ replacing $(M_\Gamma,\tilde M_\Gamma)$. Thus
    $|\hat f(\gamma)-f(\gamma)|\leq \frac{C}{T-K}.$ As a result:
    \begin{align}
        \abs*{\hat\gamma-\gamma_0} = \abs*{\hat f(\hat\gamma)-f(\gamma_0)}
        \le \abs*{\hat f(\hat\gamma)-\hat f(\gamma_0)} +\abs*{\hat f(\gamma)-f(\gamma_0)}
        \le \frac{1}{2} \abs*{\hat\gamma-\gamma_0} +\frac{C}{T-K} .
    \end{align}
    This implies $\abs*{\hat\gamma-\gamma_0}\leq \frac{1}{2}\frac{C}{T-K} =O_p(1/T).$ 
    
    Equation \eqref{eq: equivalence3} yields
    \begin{align}
    \abs*{r'(\hat \beta^\text{IV}(\hat\Gamma) -\hat \beta^\text{IV}(\Gamma_0))}
    &\le \norm*{(X'X)^{-1/2} r}\left\|(I+A_{\Gamma_0}M)^{-1}( A_{\hat \Gamma}-A_{\Gamma_0})M(I+ A_{\hat \Gamma}M)^{-1}\varepsilon\right\| \\ 
    &\leq \norm*{(X'X)^{-1/2} r}\left\|(I+A_{\Gamma_0}M)^{-1}\right\| \norm*{ A_{\hat \Gamma}-A_{\Gamma_0}} \norm*{(I+A_{\hat \Gamma}M)^{-1}} \|\varepsilon\|,
    \end{align}
    where $A_{\Gamma_0}=(I-\Gamma_0)\inverse \Gamma_0$ and $A_{\hat \Gamma}=(I-\hat\Gamma)\inverse \hat\Gamma$.  
    We have $\|\varepsilon\|=O_p(\sqrt{T})$ and $\| A_{\hat \Gamma}- A_{\Gamma_0}\| = O_p(\hat\gamma-\gamma_0)=O_p(1/T).$ By Assumption \ref{Ass: OLS}\ref{Ass: OLS 3} $r'(X'X)^{-1}r=O(1/T)$, while $\left\|(I+A_{\Gamma_0}M)^{-1}\right\|$ is uniformly bounded. Thus, $r'\hat \beta^\text{IV}(\hat\Gamma) -r'\hat \beta^\text{IV}(\Gamma_0)=O_p(T\inverse)=o_p(1/\sqrt{T})$.
\end{enumerate*}
\end{proof}

The proofs of Theorems \ref{Thm: consistency of variance}-\ref{thm: inference for Gaussian errors} follow after the proof of Theorem \ref{thm: consistency of multi-period version} as they rely on results established therein.

\begin{proof}[Proof of Theorem \ref{thm: consistency of multi-period version}.]
    Define $R_{\Gamma,\ell} = \tr[(D')^\ell (I-\Gamma) \tilde M_{\Gamma} ]$. Note that $\{r'\bar S^{-1}\alpha_\ell,R_{\Gamma,\ell}\}_{\ell=1}^L$ and $r'\hat\beta^\text{IV}(\Gamma)-r'\beta$ are invariant under the transformation $\Theta$ of Lemma \ref{lem: rotation}. Thus, we may assume without loss of generality that Lemma \ref{lem: rotation}(\ref{lem: rotation 1} and \ref{lem: rotation 2}) hold with $\Theta=I_K$ and $Z=(I-\Gamma')X$. Since $(\alpha_1,\dots,\alpha_L)$ is spanned by the first $L+1$ basis vectors, we have $M_2=\tilde M_2$, $X_2=\tilde X_2$, and $Z_2 = \tilde Z_2$, where $M_2$, $X_2$, and $Z_2$ are defined as in Lemma \ref{lem: Frisch-Waugh} using $X$, $Z=(I-\Gamma')X$, and $K_1=L+1$, while $\tilde M_2$, $\tilde X_2$, $\tilde Z_2$ are analogously defined starting from $\tilde X$, $\tilde Z=(I-\Gamma')\tilde X$, and $K_1=L+1$. This also implies that $X_2$, $Z_2$, and $M_2$ are non-random conditionally on $\tilde X$. Lemma \ref{lem: Frisch-Waugh}\ref{lem: Frisch-Waugh 1} now yields:
    \begin{align}\label{eq: FWL}
        r'\hat\beta^\text{IV}(\Gamma)-r'\beta = r'(Z_1' M_{2}X_1)^{-1}Z_1' M_{2}\varepsilon.
    \end{align}
    Defining $\bar S_2 = \E\big[ Z_1' M_2 X_1 \!\mid\! \tilde X \big]=\tilde Z_1'\tilde X_1+\sigma^2\sum_{j,\ell=1}^L \alpha_{*,j}\alpha_{*,\ell}' \tr[(D')^j(I-\Gamma) M_{2} D^\ell]$ and $R_{2,\ell}=\tr[(D')^\ell(I-\Gamma) M_{2}]$, we show as a \textbf{first} step that 
    \begin{align}\label{eq: projection second step}
       r' \hat\beta^\text{IV}(\Gamma)-r'\beta= \sigma^2 \sum_{\ell=1}^L r_*'\bar S_2\inverse \alpha_{*,\ell} R_{2,\ell}+ o_p(1).
    \end{align}
    Equation \eqref{eq: projection second step} follows from equation \eqref{eq: FWL} and the following statements proven below:
    \begin{align}
        \norm*{ (Z_1'M_2 X_1-\bar S_2)/T }_F = O_p(\sqrt{1/T}), \label{eq: denominator with projection} \\ 
        \norm*{ (\bar{S}_{2}/T)^{-1}} \le \frac{1}{(1-\|\Gamma\|)}, \label{eq: norm of denominator with projection} \\ 
        \norm*{ (Z_1' M_{2}\varepsilon -\sigma^2 \sum\nolimits_{\ell=1}^L \alpha_{*,\ell} R_{2,\ell})/T } = O_p(\sqrt{1/T}), \label{eq: numirator with projection}
    \end{align}
    and $|R_{2,\ell}|\leq T$ and thus $\|\sum_{\ell=1}^L\alpha_{*,\ell}R_{2,\ell}\|=O(T)$. Let $u_\ell = (\varepsilon_{1-\ell},\dots,\varepsilon_{T-\ell})'$ as in Lemma \ref{lem: order of terms infinite lag}. Equation \eqref{eq: denominator with projection} considers a $(L+1)\times (L+1)$ matrix with mean of zero:
    \begin{align}
        Z_1' M_{2}X_1 - \bar S_2
        &= \left(\tilde X_1+\sum_{j=1}^L u_j\alpha_{*,j}'\right)'(I-\Gamma) M_{2}\left(\tilde X_1+\sum_{\ell=1}^L u_\ell\alpha_{*,\ell}'\right)- \bar S_2 \\ 
        &= \sum_{\ell=1}^L\alpha_{*,\ell} u_\ell'(I-\Gamma)M_{2}\tilde X_1 
        + \tilde X_1'(I-\Gamma) M_{2}\sum_{\ell=1}^L u_\ell\alpha_{*,\ell}' \\
        &+\sum_{j,\ell=1}^L\alpha_{*,\ell} \alpha_{*,j}' \left [u_\ell'(I-\Gamma) M_2 u_j- \E[ u_\ell'(I-\Gamma) M_{2}u_j]\right].
    \end{align}
Notice that for $A =(I-\Gamma) M_{2} $, we have $\|A\| \le \frac{1+\|\Gamma\|}{1-\|\Gamma\|}$ and $\|A\|_F=O(\sqrt{T}).$ Thus, applying Lemma \ref{lem: order of terms infinite lag}(\ref{lem: order of terms infinite lag 1} and \ref{lem: order of terms infinite lag 2}) we obtain \eqref{eq: denominator with projection}. As Lemma \ref{lem: equivalence}\ref{lem: equivalence 1} yields that $B=(I-\Gamma) M_2+ M_2'(I-\Gamma')$ is a non-negative definite matrix, we have 
    \begin{align} 
       x'(\bar S_2+\bar S_2')x 
       &= x'\!\left[ \tilde Z_1'\tilde X_1+\tilde X_1'\tilde Z_1\right]\! x + \sum_{j,\ell=1}^L x'\alpha_{*,j} \E\left[u_j'B u_\ell|\tilde X\right] \alpha_{*,\ell}'x \\
       &\ge x'\tilde X_1'\!\left[ I - (\Gamma +\Gamma')/2\right]\!\tilde X_1 x 
       \ge (1-\norm{\Gamma})\norm{\tilde X_1 x}^2 = (1-\norm{\Gamma})T \norm{x}^2 
    \end{align}
    Applying Lemma \ref{lem: trace}\ref{lem: trace 2} now implies \eqref{eq: norm of denominator with projection}. To prove \eqref{eq: numirator with projection}, we note that
    \begin{align}\label{eq: numerator before}
        Z_1'M_2 \varepsilon = \tilde X_1'(I-\Gamma) M_2 \varepsilon + \sum_{\ell=1}^L\alpha_{*,\ell} u_\ell'(I-\Gamma) M_{2}\varepsilon.
    \end{align}
    The first term is $O_p(\sqrt{T})$ due to Lemma \ref{lem: order of terms infinite lag}\ref{lem: order of terms infinite lag 1} applied with $A=(I-\Gamma)M_2$ and $\alpha_{\ell}=0$ for $\ell > 0$. As $\sigma^2R_{2,\ell} = \E[u_\ell'(I-\Gamma) M_2\varepsilon]$, Lemma \ref{lem: order of terms infinite lag}\ref{lem: order of terms infinite lag 2} yields $\sum_{\ell=1}^L\alpha_{*,\ell} \big[u_\ell'(I-\Gamma) M_2\varepsilon-R_{2,\ell}\big] = O_p(\sqrt{T})$. This leads to \eqref{eq: numirator with projection}.

    As the \textbf{second} and final step of the proof, we show that \eqref{eq: projection second step} implies \eqref{eq: some bias}. The first difference is that \eqref{eq: projection second step} depends on $r'_*$, $\alpha_{*,\ell}$, and the $(L+1)\times (L+1)$ matrix $\bar S_2$, while \eqref{eq: some bias} is described using $r' = (r'_*,\mathbf{0}_{K-L-1}')$, $\alpha_\ell' = (\alpha'_{*,\ell},\mathbf{0}_{K-L-1}')$, and the $K\times K$ matrix $\bar S_\Gamma$. The second difference between these two statements is that \eqref{eq: projection second step} employs the oblique projection $\tilde M_{2}$, which projects off the $T \times (K-L-1)$ matrix $\tilde X_2$, while \eqref{eq: some bias} employs $\tilde M_{\Gamma}$, which projects off the full $T \times K$ matrix of regressors $\tilde X = [\tilde X_1,\tilde X_2]$. 

    For the first of these discrepancies, from Lemma \ref{lem: rotation}(\ref{lem: rotation 1} and \ref{lem: rotation 2}) we have that the lower left $(K-L-1)\times (L+1)$ blocks of $\alpha_\ell\alpha_j'$ and $\tilde Z'\tilde X$ are zero. Thus the upper left $(L+1)\times (L+1)$ block of $\bar S_\Gamma^{-1}$ is equal to the inverse of the upper left $(L+1)\times (L+1)$ block of $\bar S_\Gamma$. Thus
    \begin{align}
        r'\bar S_\Gamma \inverse \alpha_\ell = r_*'\left(\bar S_2 + \sigma^2\sum_{j,\ell=1}^L\alpha_{*,j}\alpha_{*,\ell}'\Delta_{j \ell} \right)^{-1}\alpha_{*,\ell},
    \end{align}
    where $\Delta_{j\ell}=\tr\big[(D')^j(I-\Gamma)(\tilde M_{\Gamma}-\tilde M_2)D^\ell\big]$. Now, Lemma \ref{lem: Frisch-Waugh}\ref{lem: Frisch-Waugh 3} gives $\abs{\Delta_{j \ell}} \le C\|D^j(D')^i(I-\Gamma)\| = O(1)$ so that $\sum_{\ell=1}^L (r'\bar S_\Gamma\inverse\alpha_\ell -r_*'\bar S_2\inverse \alpha_{*,\ell})R_{\Gamma,\ell}=o_p(1)$. For the second discrepancy, we have $R_{\Gamma,\ell}-R_{2,\ell} = \tr\big[(D')^\ell(I-\Gamma)(\tilde M_{\Gamma}-\tilde M_2)\big]$ and Lemma \ref{lem: Frisch-Waugh}\ref{lem: Frisch-Waugh 3} similarly yielding $\abs{R_{\Gamma,\ell}-R_{2,\ell}} = O(1)$. Thus $\sum_{\ell=1}^L (R_{\Gamma,\ell} -R_{2,\ell})r_*'\bar S_2\inverse \alpha_{*,\ell}=o_p(1)$. In conclusion, we have $\sum_{\ell=1}^L r'\bar S_\Gamma\inverse\alpha_\ell R_{\Gamma,\ell}-r_*'\bar S_2\inverse \alpha_{*,\ell} R_{2,\ell}=o_p(1)$, so \eqref{eq: projection second step} implies \eqref{eq: some bias} and Theorem \ref{thm: consistency of multi-period version} follows.
\end{proof}

\begin{proof}[Proof of Theorem \ref{Thm: consistency of variance}.]
    From equation \eqref{eq: equivalence1} we have $\hat \sigma^2(\Gamma) = \frac{\varepsilon'(I-\Gamma)M_\Gamma \varepsilon}{T-K_\Gamma}$. Note that $\sigma^2$ and $\hat \sigma^2(\Gamma)$ are invariant under the transformation $\Theta$ of Lemma \ref{lem: rotation}. Thus, we may assume without loss of generality that Lemma \ref{lem: rotation}(\ref{lem: rotation 1} and \ref{lem: rotation 2}) hold with $\Theta=I_K$. As in the proof of Theorem \ref{thm: consistency of multi-period version}, we therefore have $M_2 = \tilde M_2$. Applying Lemma \ref{lem: Frisch-Waugh}\ref{lem: Frisch-Waugh 2} yields:
    \begin{align}\label{eq: mid equation in variance proof}
        \hat \sigma^2(\Gamma)=\frac{\varepsilon'(I-\Gamma)\tilde M_2 \varepsilon}{T-K_\Gamma}-\frac{1}{T-K_\Gamma}\varepsilon'(I-\Gamma)\tilde M_2X_1(Z_1'\tilde M_2 X_1)^{-1}Z_1' \tilde M_2 \varepsilon.
    \end{align}
    Lemma \ref{lem: order of terms infinite lag}\ref{lem: order of terms infinite lag 2} applied with $A=\frac{(I-\Gamma)M_2}{T-K_\Gamma}$ implies that $\frac{\varepsilon'(I-\Gamma)M_2 \varepsilon}{\sigma^2(T-K_\Gamma)}=1+o_p(1)$. For the second term in \eqref{eq: mid equation in variance proof}, we use statements \eqref{eq: denominator with projection}--\eqref{eq: numirator with projection} and the second part of the proof of Theorem \ref{thm: consistency of multi-period version} to obtain the conclusion of Theorem \ref{Thm: consistency of variance}.
\end{proof}

\begin{proof}[Proof of Theorem \ref{thm: inference for smaller K}.]
    Let $\gamma_0$ be as in Theorem \ref{Thm: IV}. Note that $r'\hat\beta^\text{IV}(\Gamma_0), r'\beta$ and $\hat\varSigma_T(\Gamma_0)=\hat\sigma^2(\Gamma_0)\|r'(X'(I-\Gamma_0)X)^{-1}X'(I-\Gamma_0)\|^2$ are invariant under the transformation $\Theta$ of Lemma \ref{lem: rotation}. Thus, we may assume without loss of generality that Lemma \ref{lem: rotation}(\ref{lem: rotation 1} and \ref{lem: rotation 2}) hold with $\Theta=I_K$ and $Z=(I-\Gamma_0')X$. Following the proof of Theorem \ref{thm: consistency of multi-period version}, we have formula \eqref{eq: FWL}, where the denominator satisfies \eqref{eq: denominator with projection} and \eqref{eq: norm of denominator with projection}. This implies that
    \begin{align}\label{eq: first step in inference proof}
        r'\hat\beta^\text{IV}(\Gamma_0) - r'\beta 
        = \big(1+o_p(1)\big)\left( \sum_tw_t\varepsilon_t+r_*'\bar S_2^{-1}\alpha_*\varepsilon'B\varepsilon\right),
    \end{align} 
    where $w'=(w_1,\dots,w_T) = r_*'\bar S_2^{-1}\tilde Z_1'M_2$, $B=D'(I-\Gamma_0)M_2$, and $(\bar S_2, M_2)$ is defined as in the proof of Theorem \ref{thm: consistency of multi-period version} with $\Gamma=\Gamma_0$. The weights $\{w_t\}$ and the matrix $B$ are measurable with respect to $\tilde X$. Below, we show that
    \begin{align}\label{eq: CLT from Mikkel}
         \frac{r'\hat\beta^\text{IV}(\Gamma_0)-r'\beta}{\sqrt{\varSigma_T}}=\frac{\sum_tw_t\varepsilon_t+r_*'\bar S_2^{-1}\alpha_*\varepsilon'B\varepsilon}{\sqrt{\varSigma_T}}\Rightarrow N(0,1),
    \end{align}
    where $\varSigma_T=\sigma^2\sum_tw_t^2+\sigma^4(r_*'\bar S_2^{-1}\alpha_*)^2\tr(B^2+B'B)$. We have $$\varepsilon'B\varepsilon=\sum_tB_{tt}\varepsilon_t^2+\sum_{t}\sum_{s\neq t}\frac{B_{st}+B_{ts}}{2}\varepsilon_t\varepsilon_s.$$ 
    Lemma \ref{lem: conditions for CLT}(\ref{lem: conditions for CLT 1} and \ref{lem: conditions for CLT 3}) together with $K/T\to 0$ imply that $(r_*'\bar S_2^{-1}\alpha_*)^2\left(\sum_tB_{tt}\varepsilon_t^2\right)/\sqrt{\varSigma_T}\xrightarrow{p} 0$ and $\frac{\tr(B^2)}{\tr(B'B)}\to 0$.

    We obtain statement \eqref{eq: CLT from Mikkel} by establishing the four conditions, (i)--(iv), of \cite{solvsten2020robust}, Corollary A2.8, located in the Supplementary Appendix of that article. Condition (i) is automatically satisfied if we define $w_{t,T}=\frac{w_t}{\sqrt{\varSigma_T}}$, $M_{st}=\frac{r_*'\bar S_2^{-1}\alpha_*}{2\sqrt{\varSigma_T}}(B_{st}+B_{ts})$ for $s\neq t$, and $M_{tt}=0$. Condition (iv) is implied by Assumption \ref{Ass: OLS}\ref{Ass: OLS 2}. To establish condition (ii) we note that by Lemma \ref{lem: Frisch-Waugh}\ref{lem: Frisch-Waugh 1}, we have for any $\tilde r$ of the form $\tilde r' =(\tilde r_*',\mathbf{0}_{K-L-1}')$ that
    \begin{align}
        \tilde r'(\tilde Z'\tilde X)^{-1}\tilde Z'= \tilde r_*'(\tilde Z_1'M_2\tilde X_1)^{-1}\tilde Z_1'M_2
    \end{align}
    where we will use that $M_2=\tilde M_2$ since $X_2=\tilde X_2$ is strictly exogenous. Thus for the specific choice of $\tilde r$ where $\tilde r_*'=r_*'\bar S_2^{-1}(\tilde Z_1'M_2\tilde X_1)$ we have
    \begin{align}
        w' = r_*'\bar S_2^{-1}\tilde Z_1'M_2 = \tilde r_*'(\tilde Z_1'M_2\tilde X_1)^{-1}\tilde Z_1'M_2 = \tilde r' (\tilde Z'\tilde X)^{-1}\tilde Z'
    \end{align}
    so that $ w_t=\tilde r'(\tilde Z'\tilde X)^{-1}(\tilde X_t-\gamma_0\tilde X_{t+1}).$ Note, that
    \begin{align}
        \max_t|w_t| &\le \norm*{(\tilde X'\tilde X)^{1/2}(\tilde X'\tilde Z)^{-1}\tilde r}(1+|\gamma_0|)\max_t \norm*{(\tilde X'\tilde X)^{-1/2}\tilde X_t}, \\
        \sum_t w_t^2 &= \tilde r'(\tilde Z'\tilde X)^{-1}\tilde X'(I-\Gamma_0)(I-\Gamma_0')\tilde X (\tilde X'\tilde Z)^{-1}\tilde r 
        \ge (1-|\gamma_0|)^2 \norm*{(\tilde X'\tilde X)^{1/2}(\tilde X'\tilde Z)^{-1}\tilde r}^2.
        \shortintertext{Thus,}
        \max_{t} \abs{w_{t,T}} 
        &\le \frac{\max_t|w_t|}{\sqrt{\sum_t w_t^2}} 
        \le \frac{1+|\gamma_0|}{1-|\gamma_0|} \max_t\|(\tilde X'\tilde X)^{-1/2}\tilde X_t\|\to 0.
    \end{align}
    For condition (iii), we note that Lemma \ref{lem: conditions for CLT} and $K/T\to 0$ yields
    \begin{align}
        \sum_{s}\sum_{t\neq s}\left[\frac{B_{st}+B_{ts}}{2}\right]^2= \frac{1}{2}\tr(B^2+B'B)(1+o(1)).
    \end{align}
    This yields $\sum_s\sum_{t\neq s}M_{st}^2\leq 1$. Furthermore, 
    \begin{align}
        \left\|\frac{B+B'}{2}-diag(B)\right\|\leq \|B\|+\max_t|B_{tt}|=O(1).
    \end{align}
    Therefore we have $\norm{(M_{st})_{s,t}} \to 0$ and have therefore established \eqref{eq: CLT from Mikkel}.

    Finally, we prove that $\frac{\hat\varSigma_T}{\varSigma_T}\xrightarrow{p} 1$. Reusing the argument in the proof of Theorem \ref{Thm: IV}, we first have that $(\hat\varSigma_T-\hat\varSigma_T(\Gamma_0))/\varSigma_T = o(1)$. By Lemma \ref{lem: Frisch-Waugh}\ref{lem: Frisch-Waugh 2}, we have for $u = (\varepsilon_{0},\dots,\varepsilon_{T-1})'$ that
    \begin{align}
        \hat\varSigma_T(\Gamma_0) &= \hat\sigma^2(\Gamma_0)\|r_*'(Z_1'M_2X_1)^{-1}Z_1'M_2\|^2 
        = [1+o_p(1)]\sigma^2r_*'\bar S_2^{-1}Z_1'M_2M_2'Z_1(\bar S_2')^{-1}r_*\\
        &=[1+o_p(1)]\sigma^2r_*'\bar S_2^{-1}(\tilde Z_1'+\alpha_*u(I-\Gamma_0'))M_2M_2'(\tilde Z_1+(I-\Gamma_0')u\alpha_*')(\bar S_2')^{-1}r_*,
    \end{align}
    where we also used Theorem \ref{Thm: consistency of variance} and \eqref{eq: denominator with projection}. From Lemma \ref{lem: conditions for CLT}(\ref{lem: conditions for CLT 1} and \ref{lem: conditions for CLT 3}), we have 
    \begin{align}
        \E[ u'(I-\Gamma_0')M_2M_2'(I-\Gamma_0)u]
        = [1+o(1)]\E[\varepsilon'BB'\varepsilon] 
        = [1+o(1)]\sigma^2\tr[B'B].
    \end{align}
    From $K/T\to 0$ we have $\frac{\tr(B^2+B'B)}{\tr(B'B)}\to 1$ and therefore
    \begin{align}\label{eq: small terms in Sigma}
        \frac{\hat\varSigma_T \!-\! \varSigma_T}{\varSigma_T}
        = \frac{2\sigma^2r_*'\bar S_2^{-1}\alpha_*u(I \!-\! \Gamma_0)M_2M_2'\tilde Z_1 \!+\! \sigma^2(r_*'\bar S_2^{-1}\alpha_*)^2(\varepsilon'BB'\varepsilon \!-\! \E \varepsilon'BB'\varepsilon)}{\varSigma_T}+o_p(1). \qquad 
    \end{align}
    Define $R=M_2'\tilde Z_1$, $\xi_1=r_*'\bar S_2^{-1}\alpha_*u(I-\Gamma_0)M_2M_2'\tilde Z_1$, $\xi_2=(r_*'\bar S_2^{-1}\alpha_*)^2[\varepsilon'BB'\varepsilon-\E \varepsilon'BB'\varepsilon]$. Now,
    \begin{align}
        \E[\xi_1^2] &= C(r_*'\bar S_2^{-1}\alpha_*)^2\tr(R'B'BR)\leq C(r_*'\bar S_2^{-1}\alpha_*)^2\tr(R'R)\|B'B\|; \\
        \frac{\xi_1}{\varSigma_T} 
        &=O_p\left(\frac{r_*'\bar S_2^{-1}\alpha_*\sqrt{\tr(R'R)\|B'B\|}}{\varSigma_T}\right) \\
        &= O_p\left(\sqrt{\frac{\|B'B\|}{\tr(B'B)}}\frac{(r_*'\bar S_2^{-1}\alpha_*)^2\tr(B'B)+\tr(R'R)}{\varSigma_T}\right)
        =O_p\left(\sqrt{\frac{\|B'B\|}{\tr(B'B)}}\right).
    \end{align}
    Lemma \ref{lem: order of terms infinite lag}\ref{lem: order of terms infinite lag 2} yields $\varepsilon'BB'\varepsilon-\E\varepsilon'BB'\varepsilon=O_P(\|B'B\|_F) $. Thus
    \begin{align}
        \frac{ \xi_2}{\varSigma_T} = O_p\left(\frac{\|B'B\|_F}{\tr(B'B)}\right)= O_p\left(\frac{\sqrt{\tr(B'B)\|B'B\|}}{\tr(B'B)}\right)=O_p\left(\frac{\|B\|}{\sqrt{\tr(B'B)}}\right).
    \end{align}
    Thus, by Lemma \ref{lem: conditions for CLT}\ref{lem: conditions for CLT 3}, both terms in \eqref{eq: small terms in Sigma} are $O_p(1/\sqrt{T}).$
\end{proof}

\begin{proof}[Proof of Theorem \ref{thm: inference for Gaussian errors}]
    Note that $r'\hat\beta^\text{IV}(\Gamma), r'\beta$ and $\hat\varSigma_T(\Gamma)=\hat\sigma^2(\Gamma)\|r'(X'(I-\Gamma)X)^{-1}X'(I-\Gamma)\|^2$ are invariant under the transformation $\Theta$ of Lemma \ref{lem: rotation}. Thus, we may assume without loss of generality that Lemma \ref{lem: rotation}(\ref{lem: rotation 1} and \ref{lem: rotation 2}) hold with $\Theta=I_K$ and $Z=(I-\Gamma)X$. Proceeding as in the proof of Theorem \ref{thm: inference for smaller K} we arrive at equation \eqref{eq: first step in inference proof} with $B=D'(I-\Gamma)M_2$. Due to Gaussianity of the errors, $r_*'\bar S_2^{-1}\tilde Z_1'M_2\varepsilon=w'\varepsilon$ has a Gaussian distribution conditionally on $\tilde X$ with conditional variance $\sigma^2\|w\|^2=\sigma^2\|r_*'\bar S_2^{-1}\tilde Z_1'M_2\|^2$. Below we show that 
    \begin{align}\label{eq: quadratic CLT}
        \frac{\varepsilon'B\varepsilon}{\sigma^2\sqrt{\tr(B^2)+\tr(B'B)}}\Rightarrow N(0,1)
    \end{align}
    and that this term is asymptotically independent from the conditionally Gaussian term $w'\varepsilon$.

    Define $P_w=\frac{ww'}{w'w}$ and $M_w=I-P_w$ then
   $\varepsilon'B\varepsilon= 2\varepsilon'BP_w\varepsilon- \varepsilon'P_wBP_w\varepsilon+\varepsilon'M_wBM_w\varepsilon.$
    We notice that
    \begin{align} 
        \left|\varepsilon'P_wBP_w\varepsilon\right|=&\left(\frac{w'\varepsilon}{\|w\|}\right)^2\left|\frac{w'Bw}{w'w}\right|\leq \|B\|\cdot \chi^2_1=O_p(1); \\
        \varepsilon'BP_w\varepsilon=&\frac{w'\varepsilon}{\|w\|}\frac{w'B\varepsilon}{\|w\|}=\frac{w'\varepsilon}{\|w\|}N(0,\frac{w'BB'w}{w'w})=O_p(1).
    \end{align}
    Due to Lemma \ref{lem: conditions for CLT} we have $\frac{T}{\tr(B^2)+\tr(B'B)}=O(1).$ Thus
    \begin{align}
        \frac{\varepsilon'B\varepsilon}{\sqrt{\tr(B^2)+\tr(B'B)}}=\frac{\varepsilon'M_wBM_w\varepsilon}{\sqrt{\tr(B^2)+\tr(B'B)}}+o_p(1)
    \end{align}
    But $\frac{\varepsilon'M_wBM_w\varepsilon}{\sqrt{\tr(B^2)+\tr(B'B)}}$ is independent from $\frac{w'\varepsilon}{\|w\|}\sim N(0,\sigma^2)$. This implies that $\frac{\varepsilon'B\varepsilon}{\sqrt{\tr(B^2)+\tr(B'B)}}$ is asymptotically independent from the first term. Since $\tr(B)=\sum_tB_{tt}=O_p(1)$:
    \begin{align*}
        \varepsilon'B\varepsilon=\sum_{t}\sum_{s\neq t}\frac{B_{ts}+B_{st}}{2}\varepsilon_t\varepsilon_s+\sum_tB_{tt}(\varepsilon_t^2-\sigma^2) + O_p(1).
    \end{align*}
    Using that $\E\varepsilon_t^4=3\sigma^4$, one can show that 
    \begin{align}
        Var(\varepsilon'B\varepsilon)=2\sigma^4\sum_{t}\sum_{s\neq t}\left(\frac{B_{ts}+B_{st}}{2}\right)^2+2\sigma^4\sum_tB_{tt}^2=\tr(B^2)+\tr(B'B),
    \end{align}
    and due to Lemma \ref{lem: conditions for CLT} the right-hand-side grows no slower than of order $T$. Since $\max_t |B_{t,t}|\leq \|B\|=O(1)$ and the operator norm of $B+B'$ is bounded, conditions (i)--(iv) of Corollary A2.8 in \cite{solvsten2020robust} hold. Therefore \eqref{eq: quadratic CLT} holds and we have asymptotic Gaussianity.

    By the same argument as in the proof of Theorem \ref{thm: inference for smaller K}, we can show that
    \begin{align}
    \frac{\hat\sigma^2(\Gamma)\|r'(X'(I-\Gamma)X)^{-1}X'(I-\Gamma)\|^2}{\sigma^2\|r_*'\bar S_2^{-1}\tilde Z_1'M_2\|^2+\sigma^4(r_*'\bar S_2^{-1}\alpha_*)^2\tr(B'B)}\to^p 1.
    \end{align}
    Pre-multiplying this estimator by $1+\psi=\frac{|\tr(B^2)|+\tr(B'B)}{\tr(B'B)}$ guarantees that the resulting quantity asymptotically weakly exceeds $\varSigma_T=\sigma^2\sum_tw_t^2+\sigma^4(r_*'\bar S_2^{-1}\alpha_*)^2\tr(B^2+B'B)$.
\end{proof}

\section{Additional Simulations}
\label{app: simulations}

The outcome vector is generated as $y = X\beta + \varepsilon$ with $\varepsilon \sim N(0,I)$ and $\beta=0$. The design matrix is generated as $x_{1,t}=\tilde x_{1,t}+a\varepsilon_{t-1}$ and $X_{-1}=\tilde X_{-1}$, where $\tilde X$ is generated as a rotated MA(1) process with $\tilde X\tilde X'/T=I_K$, independent from $\varepsilon$. Specifically, we generate $v_t = \rho u_{t-1} + u_t$ with $\{u_t\}_{t=1}^T$ \emph{i.i.d.} $N(0,I_K)$ and define $\tilde X=V(V'V/T)^{-1/2},$ where the square root comes from Cholesky decomposition. Across simulations, we fix the sample size at $T=200$ and the coefficient on the feedback mechanism at $a=1.5$. Simulation results are summarized in Figure \ref{fig:biasMA} with the left panel showing results for number of regressors $K$ between $4$ and $150$ (fixing $\rho$ at $0.8$). The right panel reports the results for the auto-correlation in regressors $\rho$ between $0$ and $0.98$ (fixing $K$ at $50$). We report simulated values of absolute bias and standard deviation for the first coordinate of OLS and IV together with the mean absolute value of the ratio of the lower trace of $M$ to the sample size. The results are extremely similar to the results reported in Section \ref{sec: bias OLS} both in term of the size of the bias/standard deviations as well as dependence on the number of regressors and their one-period predictability. 

\begin{figure}[htbp]
\centering
\includegraphics[width=\columnwidth]{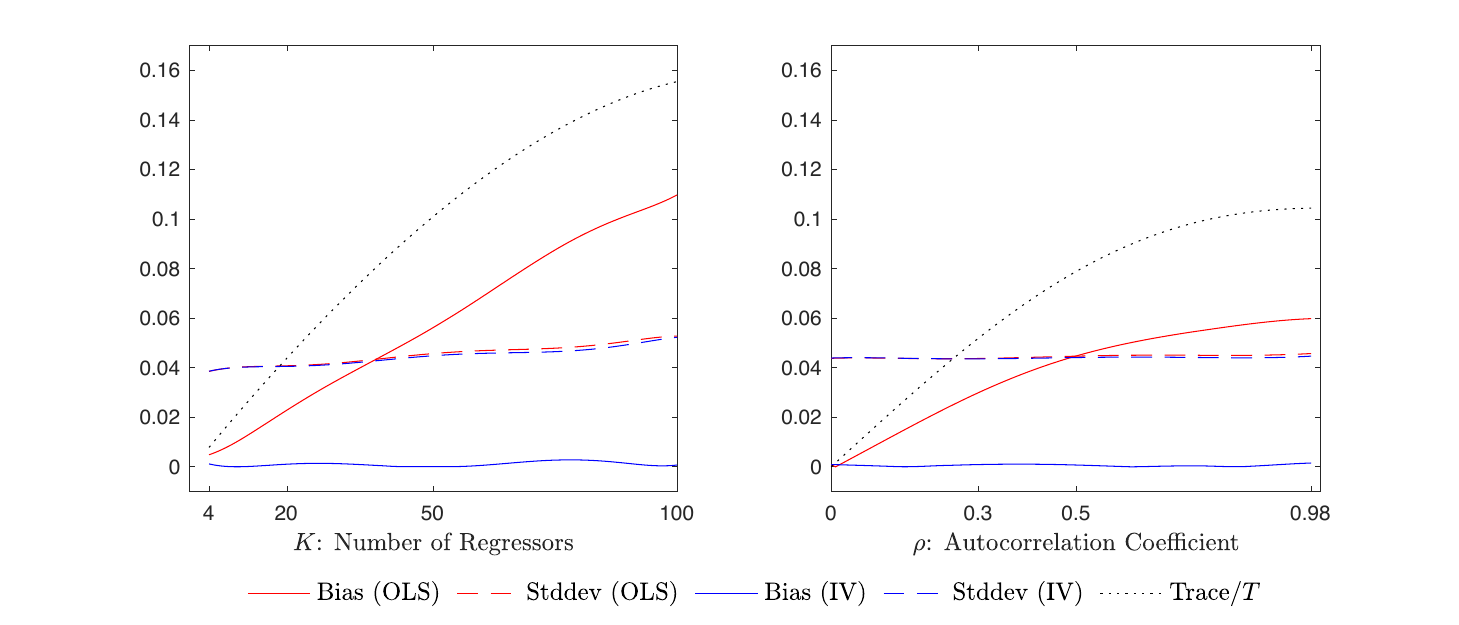}
\caption{Absolute Bias and Standard Deviation of OLS and IV with T=200}
\label{fig:biasMA}
\end{figure}

\end{appendices}

\end{document}